\documentclass[journal]{IEEEtran}

\usepackage{amsmath}
\usepackage{graphicx}

\usepackage{flushend}
 
\usepackage{latexsym}
\usepackage{amsfonts,amsthm}
\usepackage{graphicx}  
\usepackage{subfigure}
\usepackage{multirow}
\usepackage{makecell}
\usepackage{amssymb}
\usepackage{enumerate} 
\usepackage{amsmath}
\usepackage[usenames]{color}
\usepackage[dvipsnames]{xcolor}
\usepackage{longtable}
\usepackage{tabularx}
\usepackage{chngpage}
\usepackage{lipsum}  
\usepackage{booktabs}
\usepackage{subfigure}
\usepackage{verbatim}
\usepackage{xspace}
\usepackage{algorithm}
\usepackage{url}
\usepackage[normalem]{ulem}
\usepackage{tikz}
\usetikzlibrary{arrows}

\usepackage{bm}
\usepackage{pdflscape} 

\newtheorem{Teorema}{\bf Theorem}

\newtheorem{Proposicion}{\bf Proposition}

\def\x{{\mathbf x}}
\def\y{{\mathbf y}}
\def\v{{\mathbf v}}
\def\q{{\mathbf q}}
\def\r{{\mathbf r}}
\def\z{{\mathbf z}}
\def\p{{\mathbf p}}
\def\n{{\mathbf n}}

\def\A{{\mathbf A}}

\def\C{{\mathbf C}}
\def\H{{\mathbf H}}
\def\R{{\mathbf R}}
\def\Q{{\mathbf Q}}
 
\def\V{{\mathbf V}}
\def\G{{\mathbf G}}
\def\X{{\mathbf X}}
\def\Z{{\mathbf Z}}
\def\W{{\mathbf W}}

\def\z{{\mathbf z}}
\def\P{{\mathbf P}}
\def\S{{\mathbf S}}
\def\ctA{{\mathrm{ct}_{/\mathbf{A}}}}
\def\Id{{\textbf{Id}}}

\newcommand{\Sigmab}{{\bm \Sigma}}
\newcommand{\Deltab}{{\bm \Delta}}
\newcommand{\Psib}{{\bm \Psi}}
\newcommand{\Phib}{{\bm \Phi}}

\newtheorem{Algoritmo}{\em Algorithm}
\newcommand{\mub}{\boldsymbol\mu}

\newcommand{\nub}{{\boldsymbol\nu}}

\newcommand{\meanPrior}{\mub}
\newcommand{\covPrior}{\Sigmab}
\newcommand{\meanFilt}{\mub} 
\newcommand{\covFilt}{\Sigmab}
\newcommand{\meanPred}{\mub}
\newcommand{\covPred}{\Sigmab}

\newcommand{\meanSmooth}{\mub^{\text{s}}}
\newcommand{\covSmooth}{\Sigmab^{\text{s}}}

\newcommand{\loss}{\mathcal{L}}

\newcommand{\K}{{\bf K}}

\def\gauss{{\mathcal N}}

\def\Ab{{\mathbf A}}

\def\Gb{{\mathbf G}}
\def\Qb{{\mathbf Q}}

 \newcommand{\Real}{\mathbb{R}}
 \newcommand{\Complex}{\mathbb{C}}

\newcommand{\cred}{\textcolor{black}}

\newcommand{\cblue}{\textcolor{black}}

\begin{document}

\title{Graphical Inference in Linear-Gaussian \\State-Space Models}

\author{V\'ictor~Elvira,~\IEEEmembership{Senior~Member,~IEEE,}
        and~\'Emilie Chouzenoux,~\IEEEmembership{Senior~Member,~IEEE}
\thanks{V. Elvira is with the School of Mathematics at the University of Edinburgh (UK) and The Alan Turing Institute (UK). \'E. Chouzenoux is with Universit\'e Paris-Saclay, Inria, CentraleSup\'elec, Centre de Vision Num\'erique (France).}
\thanks{V.E. and \'E.C. acknowledge support from the \emph{Agence Nationale de la Recherche} of France under PISCES (ANR-17-CE40-0031-01) and MAJIC (ANR-17-CE40-0004-01) projects. V.E. acknowledges support from the Leverhulme Research Fellowship (RF-2021-593). \'E.C. acknowledges support from the European
Research Council Starting Grant MAJORIS ERC-2019-STG-850925.}
}

\markboth{IEEE TRANSACTIONS ON SIGNAL PROCESSING}%
{ }

\maketitle

\begin{abstract}

State-space models (SSM) are central to describe time-varying complex systems in countless signal processing applications such as remote sensing, networks, biomedicine, and finance to name a few. Inference and prediction in SSMs are possible when the model parameters are known, which is rarely the case. The estimation of these parameters is crucial, not only for performing statistical analysis, but also for uncovering the underlying structure of complex phenomena. In this paper, we focus on the linear-Gaussian model, arguably the most celebrated SSM, and particularly in the challenging task of estimating the transition matrix that encodes the Markovian dependencies in the evolution of the multi-variate state. We introduce a novel perspective by relating this matrix to the adjacency matrix of a directed graph, also interpreted as the causal relationship among state dimensions in the Granger-causality sense. Under this perspective, we propose a new method called GraphEM based on the well sounded expectation-maximization (EM) methodology for inferring the transition matrix jointly with the smoothing/filtering of the observed data. We propose an advanced convex optimization solver relying on a consensus-based implementation of a proximal splitting strategy for solving the M-step. This approach enables an efficient and versatile processing of various sophisticated priors on the graph structure, such as parsimony constraints, while benefiting from convergence guarantees. We demonstrate the good performance and the interpretable results of GraphEM by means of two sets of numerical examples. 
  
\end{abstract}

\begin{IEEEkeywords}
State-space modeling, graphical inference, sparsity, proximal methods, primal-dual algorithms, Kalman filtering, EM algorithm.
\end{IEEEkeywords}

\section{Introduction}
\label{sec_intro}

State-space modeling is widely used to describe complex systems in applications of science and engineering \cite{hamilton1994state,kim1999state,Sarkka}. These discrete-time models are described by a hidden (or latent) state that evolves in a Markovian manner over time through arbitrarily complicated dynamics, which allows for a realistic modeling of complex phenomena. The observations are sequentially collected and linked to the hidden states. 

This modeling aims at mimicking complex dynamic systems in an accurate manner through a hidden latent process, which sometimes is of reduced dimension w.r.t. the multivariate time series. Alternatively, the state can be very high-dimensional and can be interpreted, e.g., each dimension of the state represents a physical magnitude in a set of 3D points, but observation cannot be acquired in all locations. This scenario is common in complex systems, which are usually composed of many simpler units. Interestingly, in those systems, each unit usually interacts with very few others \cite{watts1998collective}. For instance, the evolution of the atmosphere can be modeled with a hidden state that captures physical properties at millions of geographical locations, but from one time step to the next one, each location is only affected by few close locations \cite{bauer2015quiet,choi2021short}. Therefore, accurate and efficient inference requires realistic modeling (e.g., high-dimensional state) combined with the incorporation of prior knowledge of the inner structure of the system (e.g., sparsity in the way the dimensions of the state interact). 
This paper focuses on the relevant linear-Gaussian state-space model (LG-SSM). This model allows for exact inference, when the model parameters are known, through the Kalman filter and the Rauch-Tung-Striebel (RTS) smoother \cite[Chapter 8]{Sarkka}.
In nonlinear and/or non-Gaussian models, the inference is generally done via particle filtering \cite{djuric2003particle,doucet2009tutorial}, e.g., the BPF \cite{Gordon93}, APF \cite{pitt1999filtering}, IAPF \cite{elvira2018insearch}, and OAPF \cite{branchini2021optimized}  algorithms (see \cite{elvira2019elucidating} for a further discussion). In all these models, the parameter estimation is generally done via particle-based methods (see for instance \cite{andrieu2010particle}).

\noindent\textbf{Existing methods in the literature.} SSMs have shown to be powerful mathematical models for time series analysis. A few alternative methods to SSMs exist, e.g., classical  multivariate time series analysis  models \cite{reinsel1997elements} or polynomial data fitting for trajectory estimation in tracking applications \cite{li2016fitting,li2018joint} (see further discussion in \cite[Chapter 1]{Sarkka} or \cite{fasiolo2016comparison}). 
In the case of SSMs, and more particularly in the LG-SSM, the estimation of the model parameters is essential to tackle problems that otherwise would be unapproachable, allowing for the estimation of the mean and covariance of the hidden state through Kalman filtering and RTS smoothing. Existing methods for the estimation of model parameters in LG-SSM focus on the maximum-likelihood (ML) estimate. Two main classes of methods have been proposed in the literature \cite[Chap.~12]{Sarkka}. The first class of methods makes use of the so-called sensitivity equations \cite{Gupta1974}, or on the Fisher's identity \cite{Segal88,Segal89} (see the discussion in \cite[Sec.~10.2.4]{Cappe2005} for connections between both strategies), to evaluate efficiently the first and second derivatives of the likelihood function with respect to the unknown parameters. This allows to apply iterative optimizers, such as quasi-Newton \cite{Olsson} or Newton-Raphson \cite{Gupta1974}, to obtain the ML estimate. The second class of methods relies on the expectation-minimization (EM) algorithm \cite{shumway1982approach}\cite[Sec.~10.4]{Cappe2005}\cite[Sec.~12.2.3]{Sarkka}, where the maximization of the marginal likelihood is indirectly performed by iteratively maximizing (M-step) an expectation (E-step) of the log-likelihood. Applications of the EM strategy in the LG-SSM to various fields, e.g., finance, electrical engineering, and radar, can be found for instance in \cite{Sharma,SharmaRDL,Frenkel}. The main advantage of EM in this context is its simplicity in the implementation and the convergence stability, inherited from the EM machinery \cite{Dempster,Wu}. We refer to \cite[Sec.~1]{shumway1982approach} for a detailed discussion of the benefits and drawbacks of each class of methods. However, none of the aforementioned methods allow to compute a maximum a posteriori (MAP) estimate of the parameters in the LG-SSM. It is possible to design naive extensions by simply incorporating a prior term on the function to maximize. However, the specific strategies cannot cope with complicated prior terms. More precisely, the methods of the first class are limited to differentiable penalty terms, preventing the use of sparsity enhancing functions and constraints, which are of high interest in this context (see the discussion below). In the case of the EM algorithm,  the M-step has a closed form  for very limited priors (e.g., Gaussian). It gets intractable for most priors of interest and thus the existing framework of \cite{shumway1982approach} does not lead to any directly implementable algorithmic solution.

SSMs are powerful mathematical tools for forecasting and also bring interpretability about the hidden process, allowing to understand the uncovered relations in the state space. In this line, graphical modeling methods for time series have been proposed~\cite{Eichler2012,Bach04,Barber10}. Such representation of multivariate sequences interactions has applications in several domains such as biology \cite{Pirayre,Luengo19}, social network analysis \cite{Ravazzi}, and neuroscience \cite{Richiardi}. Graphical modeling often requires the introduction of sparsity priors  to meet interpretability and compactness (see for instance the celebrated graphical lasso approach \cite{Friedman07}). Spectral constraints (e.g., low rank) might also be useful to enhance clustering effects on the graphs \cite{chiu2021lowrank}. In both cases, this yields complicated MAP formulations, involving non differentiable terms, for which available methods for LG-SSM parameter estimations cannot be applied, as discussed in the previous paragraph.  
  

\noindent\textbf{Contributions.} {In this paper, we propose a novel framework called GraphEM for the estimation of model parameters in the LG-SSM, using prior knowledge. 
While the proposed methodology can be adapted to estimate all model parameters, here we are explicit on the estimate of the transition matrix of the SSM}, which is arguably the most complicated parameter to be estimated because (a) it is high-dimensional, (b) it intervenes in the auto-regressive process of the hidden state that cannot be observed, and (c) it is highly related to the inner structure of the complex system, requiring the incorporation of suitable prior knowledge. In the spirit of modeling complex systems as described above, the transition matrix is here supposed to be sparse.  GraphEM brings a new perspective in state-space modeling to interpret the interactions of the state dimensions between consecutive time steps as a sparse {directed} graph, encoding relations among the dimensions of the hidden state.  
\cblue{Namely, the hidden process follows an order-one auto-regressive process. We interpret the sparse transition matrix of the multi-variate process in a Granger causality manner \cite{granger1969investigating}. In particular, Granger causality (also called as predictive causality) is often considered, not as a true type of causality, but just as a metric of how well one time series allows to forecast a second one \cite{eichler2013causal}. From that perspective, the $i,j$ entry in the transition matrix in LG-SSMs encodes the weight in which the $j$-th time series in the hidden state affects the $i$-th time series in the next time step, being zero if it does not have any effect. Thus, a zero in the $i,j$ entry can be interpreted as if the $j$-th time series does not provide any further information to predict the $i$-th time series (given the other time series). In GraphEM, we allow for a variety of sparsity constraints in the transition matrix, accounting for realistic modeling in a plethora of applications. We discuss particular examples and  provide simulations both in controlled scenarios and in a wireless communication problem.}

GraphEM belongs to the family of EM algorithms \cblue{for MAP estimation}, alternating between an expectation (E)-step based on the RTS smoother that builds a majorizing function of the posterior distribution of the unknown given the data, and a sophisticated maximization (M)-step in which this function is maximized w.r.t. the unknown parameter. The proposed GraphEM algorithm involves novel methodology to incorporate realistic prior knowledge about the dynamic system, such as sparsity constraints. Specifically, the inclusion of non-Gaussian and possibly non-smooth priors requires the development of a new tailored optimization procedure in the M-step. We propose to address this challenge by resorting to a proximal primal-dual splitting methodology, that we design to suitably incorporate the desired priors. In a nutshell, the contributions of this paper are as follows:\footnote{A limited version of this work was presented by the authors in the
conference paper \cite{Chouzenoux2020}.}
\begin{itemize}
	\item Proposition of a novel graphical interpretation of the transition matrix within LG-SSMs based on (sparse) causal interactions among state dimensions in the Granger sense,
	\item \cblue{Derivation of an EM-based algorithm for computing a MAP estimate of} this matrix, with strong theoretical guarantees, 
	\item Design of a convergent convex optimization procedure for an efficient implementation of the M-step, able to account for a wide class of priors on the interpreted graph,
	\item {Presentation of two challenging numerical examples, namely a controlled scenario, and a problem of channel tracking in wireless communications. Various setups of sparse transition matrices and priors are tested, including block-sparsity penalties and nuclear norms constraints.}
\end{itemize}

The rest of the paper is structured as follows. Section~\ref{sec_back}  describes the model, the filtering/smoothing algorithms, and the background about the EM framework. The novel GraphEM algorithm is presented in Section~\ref{sec:proposed}, with a detailed explanation of the E-step, the proposed new optimization methodology for the M-step, \cblue{a discussion on the priors, both from the application and the methodological perspectives, and a convergence theorem}. The paper concludes with two numerical examples in Section~\ref{sec:experiment} and some concluding remarks in Section~\ref{sec:conclusion}.

\section{Background}
\label{sec_back}

\subsection{Notation}

We denote by $\|\x\|_2={\sqrt{\x^\top \x}}$ the Euclidean norm of $\x \in \Real^N$, where $\top$ states from the transpose operation and $\Real^N$ is the $N$-dimensional Euclidean space. We also introduce $\| \mathbf{X} \|_F$ and $\| \mathbf{X} \|_2$, the Frobenius norm and spectral norms (i.e., largest singular value), respectively, of elements $\mathbf{X} = (X(n,\ell))_{1 \leq n \leq N, 1 \leq \ell \leq M} \in \Real^{N \times M}$. $\Id_{N}$ is the identity matrix of $\Real^N$ and $\text{tr}$ is the trace operator. Bold symbols are used for matrix and vectors. The useful definitions of convex analysis are reminded on-the-fly throughout the paper. For these concepts, we rely on the notation in the textbook \cite{bauschke2017convex}. Furthermore, we introduce the shorter notation $\ctA$, for any constant independent from a variable $\A$. {Finally, given a sequence of elements $\{\x_k\}_{k=1}^K$ of length $K \geq 1$, we use the notation $\x_{k_1:k_2}$ to refer to the subsequence $\{\x_k\}_{k=k_1}^{k_2}$, for $1 \leq k_1 < k_2 \leq K$.}  

\subsection{Linear state-space model}
We consider the LG-SSM described, for $k=1,\ldots,K$, as
 \begin{align}
       \x_{k} &= \A \x_{k-1} + \q_{k},       \label{eq_model_state}
\\
     \y_k &= \H_k \x_k + \r_{k}, \label{eq_model_obs}
    \end{align}
where, 
\begin{itemize}
    \item $\{\x_{k} \}_{k=1}^K\in \mathbb{R}^{N_x}$ and  $\{\y_{k} \}_{k=1}^K\in \mathbb{R}^{N_y}$, are the hidden state and the observations, respectively, at each time $k$,
    \item $\A \in \mathbb{R}^{N_x \times N_x}$ is the transition matrix that we aim at estimating,
    \item $\{\H_k\}_{k=1}^K \in \mathbb{R}^{N_y \times N_x}$ are the observation model matrices, possibly varying with $k$, and are assumed to be known,
    \item $\{ \q_k \}_{k=1}^K \sim \mathcal{N}(0,\cred{\Q})$ is the i.i.d. state noise process, assumed to follow a zero-mean Gaussian model with known symmetric definite positive (SDP) covariance \cred{matrix $\Q \in \Real^{N_x \times N_x}$},
    \item $\{ \r_k \}_{k=1}^K \sim \mathcal{N}(0,\R_k)$ is the i.i.d. observation noise process, again zero-mean Gaussian with known SDP covariance matrices $\R_k \in \Real^{N_y \times N_y}$.
\end{itemize}
We assume an initial state distributed such that $\x_0 \sim \mathcal{N}(\x_0 ; \meanPrior_0, \covPrior_0)$ with known $\meanPrior_0 \in \mathbb{R}^{N_x}$ and SDP $\covPrior_0\in \Real^{N_x \times N_x}$. The state and the observation noises are mutually independent and also independent of the initial state~$\x_0$.


\subsection{Kalman filtering and smoothing}

In many applications (e.g., tracking), the goal is in the estimation of the hidden state $\{\x_{k} \}_{k=1}^K$ from observations $\{\y_{k} \}_{k=1}^K$. In the Bayesian/probabilistic setting, this translates into the computation, for every $k \in \{1,\ldots,K\}$, of the posterior distribution of $\x_k$. If one conditions on all observations available up to time $k$, $\y_{1:k} = \{ \y_j \}_{j=1}^k$, then the posterior probability density function (pdf), $p(\x_k|\y_{1:k})$, is the \emph{filtering} distribution. Conditioning on the whole set of observations $\y_{1:K}$, the posterior $p(\x_k|\y_{1:K})$ is the \emph{smoothing} distribution. 

Estimating the filtering and smoothing distributions is in general a challenging problem, since obtaining these distributions of interest is possible only in few models of interest. For instance, for the LG-SSM described in \eqref{eq_model_state}-\eqref{eq_model_obs}, it is possible to obtain the filtering and smoothing distributions, for $k=1,\ldots,K$, in the case where the model parameters $\A$, \cred{$\Q$}, $\{\H_k \}_{k=1}^K$, and $\{\R_k \}_{k=1}^K$ are known. Interestingly, these distributions can be obtained in an efficient sequential manner via the Kalman filter \cite{Kalman60} and the RTS smoother \cite{Briers05}. Algorithm \ref{alg_kf} describes the Kalman filter while Algorithm~\ref{alg_rts} describes the RTS smoother.

\begin{table}[!t]
\vspace{4mm}
    \centering
    \begin{tabular}{|p{0.95\columnwidth}|}
    \hline
\begin{Algoritmo}
\label{alg_kf}
Kalman Filter
\begin{enumerate}
  \item[] \textbf{Input.} Prior parameters $\meanPrior_0$\;and  $\covPrior_0$; model parameters $\A$, $\cred{\Q}$, $\{\H_k \}_{k=1}^K$, and $\{\R_k \}_{k=1}^K$; set of observations $\{\y_k \}_{k=1}^K$.
  
\item[] \textbf{Recursive step.} For $k=1,\ldots,K$  
\begin{enumerate}
  \item {{\sf Prediction/propagation step.}} 
\begin{eqnarray}
\meanPred_{k|k-1} &=& \A\meanFilt_{k-1} \\
\covPred_{k|k-1} &=& \A\covFilt_{k-1}\A^\top + \cred{\Q}
\label{eq_propagation_kf}
\end{eqnarray}
\item {{\sf Update step.}} 
\begin{eqnarray}
\nub_{k} &=& \H_k \meanPred_{k|k-1}\\
\v_{k} &=& \y_k - \nub_k \\
\S_{k} &=&  \H_k \covPred_{k|k-1} \H_k^\top + {\R_k} \\
\K_{k} &=&  \covPred_{k|k-1} \H_k^\top \S_k^{-1} \\
\meanFilt_{k} &=& \meanPred_{k|k-1} + \K_k \v_k\\
\covFilt_k &=& \covPred_{k|k-1} - \K_{k} \S_k \K_{k}^\top 
\label{eq_update_kf}
\end{eqnarray}
\end{enumerate}
\item[] \textbf{{Output.}} $\{\meanFilt_k, \covFilt_k \}_{k=1}^K$. Then, for each $k=1,...,K$: 
\begin{itemize}
  \item state filtering pdf: $p(\x_k|\y_{1:k}) = \gauss(\x_k ; \meanFilt_k, \covFilt_k)$
  \item observation predictive pdf: $p(\y_k|\y_{1:k-1}) = \gauss(\y_k;\nub_k,\S_k)$
  
\end{itemize}
\end{enumerate}
\end{Algoritmo}\\
        \hline
\end{tabular}
\end{table}

\begin{table}[!t]
\vspace{4mm}
    \centering
    \begin{tabular}{|p{0.95\columnwidth}|}
    \hline
\begin{Algoritmo}
\label{alg_rts}
RTS Smoother
\begin{enumerate}
  \item[] \textbf{Input.} Filtering parameters $\{\meanFilt_k, \covFilt_k \}_{k={0}}^K$ from the Kalman filter; model parameters $\A$ and \cred{$\Q$}.
\item[] \textbf{Initialization.} Set $\meanSmooth_{K} = \meanFilt_{K}$ and $\covSmooth_{K} = \covFilt_{K}$. 
 
\item[] \textbf{Recursive step.} For $k=K,K-1,...,{0}$  

{\begin{align}
\meanFilt_{k+1}^{-} &= \A\meanFilt_{k} \\
\covFilt_{k+1}^{-} &= \A\covFilt_{k}\A^\top + \cred{\Q}\\
\G_k &= \covFilt_{k}\A^\top \Big(\covFilt_{k+1}^{-} \Big)^{-1}\\
\meanSmooth_{k} &= \meanPred_{k|k-1} + \G_k \left(\meanSmooth_{k+1}  - \meanFilt_{k+1}^{-} \right)\\
\covSmooth_k &= \covPred_{k|k-1} - \G_k \left(\covSmooth_{k+1}  - \covFilt_{k+1}^{-} \right)\G_k^\top
\end{align}} 

\item[] \textbf{{Output.}} $\{\meanSmooth_k, \covSmooth_k \}_{k=1}^K$. Then, for each $k=1,...,K$: 
\begin{itemize}
  \item state smoothing pdf: {$p(\x_k|\y_{1:K}) = \gauss(\x_k ; \meanSmooth_k, \covSmooth_k)$} 
  
\end{itemize}
\end{enumerate}
\end{Algoritmo}\\
        \hline
\end{tabular}
\end{table}

 \subsection{EM framework for parameter estimation}
\label{sec:BackEM}

In this paper, we consider the more challenging setting in which some parameters of the LG-SSM are unknown, and must be estimated jointly with the hidden states inference.  
The problem of parameter estimation in SSM has been widely studied in the literature. Three main types of methods can be distinguished, namely (i) expectation-maximization (EM) algorithms \cite{shumway1982approach,Thiesson04,Sharma}, (ii) optimization-based methods \cite{Olsson}, and (iii) Monte Carlo methods \cite{Kantas09,luengo2020survey}. In the context of LG models, the EM strategy is particularly well suited, since it keeps a reduced computational cost while preserving part of the Bayesian interpretation~\cite{shumway1982approach}. 
\cblue{We now describe the rationale of applying the EM strategy for the estimation of the state matrix $\A$ in the LG-SSM of \eqref{eq_model_state}-\eqref{eq_model_obs}.} 
{In such context, the maximum likelihood (ML) estimate of $\A$ is not available in a closed form \cite{shumway1982approach}. Moreover, the maximum a posteriori (MAP) estimate approach is also intractable and remains unexplored to the best of our knowledge.}  

The MLEM algorithm builds iteratively an ML estimate of the LG-SSM parameters through the resolution of surrogate problems constructed following a majorization principle \cite{MoonEM}. For the sake of clarity, we describe here the resulting MLEM procedure for the LG-SSM case. Note that we focus here on the estimation of $\A$, though the MLEM for LG-SSM was initially introduced in~\cite{shumway1982approach} for estimating the state/observation and covariance noise matrices. In the sequel, we will denote $(\Ab^{(i)})_{i \in \mathbb{N}}  \in \mathbb{R}^{N_x \times N_x}$ the sequence of MLEM iterates, whose construction will be specified \cblue{below}. For every $\x_{0:K}$ with non zero probability, the log-likelihood function is
\begin{equation}
\log p(\y_{1:K} | \Ab) =  \log p(\x_{0:K},\y_{1:K} | \Ab) - \log p(\x_{0:K}|\y_{1:K},\Ab). \label{eq:ML1}
\end{equation}
\cblue{This function is continuously differentiable for $\Ab \in \mathbb{R}^{N_x \times N_x}$. Moreover, \eqref{eq:ML1} can be easily evaluated, using its recursive form:
\begin{equation}
\log p(\y_{1:K} | \Ab) =  \sum_{k=1}^K \frac{1}{2} \log | 2 \pi \S_k| + \frac{1}{2} \v_k^\top \S_k^{-1} \v_k, \label{eq:ML2}
\end{equation}
where $(\v_k,\S_k)_{1 \leq k \leq K}$ are obtained by the RTS Alg.~\ref{alg_rts} run for a given transition matrix $\Ab$. The gradient and Hessian of \eqref{eq:ML1} can also be evaluated with recursive formula (using, for instance, Fisher's identity \cite[Chap.~12]{Sarkka}). These properties are at the core of the optimization-based estimation methods in \cite{Gupta1974,Olsson}, unfortunately presenting an unstable behaviour, mostly due to the non-convexity of~\eqref{eq:ML1}. The MLEM algorithm \cite{shumway1982approach} proceeds differently, by building sequential lower bounds of \eqref{eq:ML1}, as we describe below.} Let $i \in \mathbb{N}$, associated with the current parameter estimate $\Ab^{(i)}$. We can take the expectation of \eqref{eq:ML1} over all possible values of the unknown state given $\Ab^{(i)}$, by multiplying both sides of \eqref{eq:ML1} by $ p(\x_{0:K}|\y_{1:K},\Ab^{(i)})$ and integrating over all states. Since $\int \log p(\y_{1:K} | \Ab) p(\x_{0:K}|\y_{1:K},\Ab^{(i)}) d\x_{0:K}  = \log p(\y_{1:K} | \Ab) $ (i.e., integration of a constant quantity), then
\small{
\begin{align}
& \log p(\y_{1:K} | \Ab) = \underbrace{\int p(\x_{0:K}|\y_{1:K},\Ab^{(i)}) \log p(\x_{0:K},\y_{1:K} | \Ab) d\x_{0:K}}_{\triangleq q(\Ab;\Ab^{(i)})} \nonumber \\
&\quad \underbrace{- \int p(\x_{0:K}|\y_{1:K},\Ab^{(i)}) \log p(\x_{0:K}|\y_{1:K},\Ab) d\x_{0:K}}_{\triangleq h(\Ab;\Ab^{(i)})}.
\label{eq:Eproof}
\end{align}
}
\normalsize
The latter equation holds for any $\Ab \in \mathbb{R}^{N_x \times N_x}$, including $\Ab = \Ab^{(i)}$, i.e.,
\begin{equation}
\log p(\y_{1:K} | \Ab^{(i)}) = q(\Ab^{(i)};\Ab^{(i)}) + h(\Ab^{(i)};\Ab^{(i)}). \label{eq:EMtang}
\end{equation}
Subtracting \eqref{eq:EMtang} from \eqref{eq:Eproof} yields
\begin{multline}
\log p(\y_{1:K} | \Ab) - \log p(\y_{1:K} | \Ab^{(i)}) \\
= q(\Ab;\Ab^{(i)}) - q(\Ab^{(i)};\Ab^{(i)})
+  h(\Ab;\Ab^{(i)}) - h(\Ab^{(i)};\Ab^{(i)}).
\end{multline}
Since the entropy is upper-bounded by the cross-entropy w.r.t. any other pdf (\emph{Gibb's inequality}),
 \begin{equation}
h(\Ab;\Ab^{(i)}) \geq h(\Ab^{(i)};\Ab^{(i)}),
\end{equation}
where the equality holds if and only if $\Ab = \Ab^{(i)}$. We can thus conclude that
\begin{equation}
\log p(\y_{1:K} | \Ab) - \log p(\y_{1:K} | \Ab^{(i)}) \geq q(\Ab;\Ab^{(i)}) - q(\Ab^{(i)};\Ab^{(i)}),
\end{equation}
that is,
\begin{equation}
\log p(\y_{1:K} | \Ab) \geq q(\Ab;\Ab^{(i)}) + \ctA. \label{eq:EMmaj}
\end{equation}
Again, the equality holds in \eqref{eq:EMmaj} if and only if $\Ab = \Ab^{(i)}$. Inequality \eqref{eq:EMmaj} is the cornerstone of the MLEM algorithm, which follows a majoration-minimization (MM) principle~\cite{hunter2004tutorial}. At each iteration $i \in \mathbb{N}$ of the MLEM method, the E-step computes the following expectation:
\begin{equation}
\small{
q(\Ab;\Ab^{(i)}) = \int p(\x_{0:K} | \y_{1:K},\Ab^{(i)}) \log p(\x_{0:K},\y_{1:K} | \Ab) \rm{d} \x_{0:K},
}
\label{eq:funQ0}
\end{equation}
satisfying \eqref{eq:EMmaj}. The M-step aims at maximizing $q(\Ab;\Ab^{(i)})$ with respect to $\Ab$, yielding $\Ab^{(i+1)}$. Thus, by construction, 
\begin{align}
\log p(\y_{1:K} | \Ab^{(i+1)}) & \geq q(\Ab^{(i+1)}  ;\Ab^{(i)}) + \ctA \\
& \geq q(\Ab^{(i)} ;\Ab^{(i)}) + \ctA \\
& = \log p(\y_{1:K} | \Ab^{(i)}).
\end{align}
\cblue{The MLEM guarantees the increase of the log-likelihood loss $\log p(\y_{1:K} | \Ab^{(i)})$  along iterations, which is equivalent to an increase of the ML loss~\cite{MoonEM,Dempster}. As shown in \cite{shumway1982approach}, the integral in \eqref{eq:funQ0} can be expressed as byproducts of the RTS smoother. This leads to the construction of an MLEM method to derive estimates of the parameters of an LG-SSM, jointly with the hidden states inference task. However, the aforementioned work did not include any prior knowledge on the parameters. Moreover, although the convergence of generic EM schemes has been established in \cite{Wu}, the required assumptions are not met in the case of the MLEM scheme for LG-SSM from \cite{shumway1982approach}, mostly due to the intricate recursive form of the ML loss. Finally, the derivations in \cite{shumway1982approach} were restricted to the case of constant \cred{matrices $\R$ and $\H$ in the observation model equation \eqref{eq_model_obs}.}}

\section{The GraphEM algorithm}
\label{sec:proposed}
In this section, we present a generalized version of this EM approach, able to encompass {time-varying \cred{observation model} as well as to \cblue{yield a MAP estimate of LG-SSM transition matrix, for a large class of priors}. We explicit both the E and M steps, and introduce a novel efficient iterative solver for performing the latter \cblue{with assessed convergence guarantees}. \cblue{We show the convergence of the resulting EM-based approach under reasonable assumptions.}

\subsection{Summary of GraphEM}
In this section, we present a general framework for the estimation of the transition matrix $\A$ of the state model in Eq. \eqref{eq_model_state} under suitable prior assumption. This allows to integrate useful sparsity and spectral constraints on $\A$, with the aim of promoting the interpretability and the stability of the inferred LG-SSM. These constraints are encoded in the prior distribution $p(\A)$, as we discuss in Section \ref{sec_prior}. GraphEM aims at providing the maximum a posteriori (MAP) estimator of $\A$.  More specifically, let us denote the posterior of the unknown parameter, $p(\A|\y_{1:K})$, where the hidden states have been marginalized. It is direct to show, using Bayes rule and the (strictly increasing) logarithmic function, that the maximum of $p(\A|\y_{1:K}) \propto p(\A)p(\y_{1:K}|\A)$ coincides with the minimum of the loss function
\begin{align}
(\forall \A \in \mathbb{R}^{N_x \times N_x}) \quad 
\loss_K(\A) & \triangleq \loss_0(\A)
  \cblue{+\loss_{1:K}(\A)},
    \label{eq:phik}
    \end{align}
where we denote the regularization function as 
\begin{equation}
(\forall \A \in \mathbb{R}^{N_x \times N_x}) \quad \loss_0(\A) \triangleq - \log p(\A), \label{eq:priorgen}
\end{equation}
\cblue{
and the neg-log-likelihood as
\begin{equation}
(\forall \A \in \mathbb{R}^{N_x \times N_x}) \quad \loss_{1:K}(\A) \triangleq - \log p(\y_{1:K}| \A), \label{eq:loglossl}
\end{equation}
}
with $\log p(\y_{1:K}| \A)$ defined in Eq.~\eqref{eq:ML1}. As commented above, it is not straightforward to find a minimizer of \eqref{eq:phik}, even for the case without regularization function. 

 The proposed algorithm, GraphEM, is summarized in Algorithm~\ref{algo:GRAPHEM}. GraphEM is a type of expectation-maximization (EM) method that runs for several iterations, alternating between the expectation (E)-step and the maximization (M)-step. {The E-step can be seen as a generalization of the one in MLEM from \cite{shumway1982approach}, to the case of time-varying \cred{observation matrices $\{\R_k\}_{k=1}^K$ and $\{\H_k\}_{k=1}^K$}. Moreover, it also accounts for a prior term on $\A$ (see Section \ref{sec_prior}), so as to reach a MAP estimate of the  transition matrix $\A$. The M-step thus becomes much more intricate than in the aforementioned MLEM. In particular, no close form is longer available for the update of the transition matrix. To overcome this challenge, we propose an iterative solver with sound convergence guarantees, relying on modern tools from convex analysis (see Section \ref{sec_m_step}).}  

 At each iteration $i \in \mathbb{N}$, the expectation function $q(\A;\A^{(i)})$ given in \eqref{eq:funQ0} is first computed in the E-step. This function is created by running the Kalman filter followed by the RTS smoother with the state matrix set to the estimate of the previous iteration, i.e.,  equals to $\A^{(i)}$. We then construct
\begin{align}
\mathcal{Q}(\A;\A^{(i)}) \triangleq -q(\A;\A^{(i)}) + \loss_0(\A) + \ctA, 
\label{eq:FuncQp}
\end{align}
a majorizing approximation of the MAP loss function \eqref{eq:phik}. Then, a new estimate of the transition matrix, $\A^{(i+1)}$, is obtained in a corrected M-step, as the minimizer of the regularized surrogate in \eqref{eq:FuncQp}. As we will show in Section \ref{sec:conver}, GraphEM aims at providing ultimately an estimate of the maximum of $p(\A|\y_{1:K})$, i.e., the MAP estimate of $\A$. The last iteration of GraphEM also provides, as a byproduct, the filtering and smoothing distribution, given this last version of the transition matrix.

\begin{table}[!t]
\vspace{4mm}
    \centering
    \begin{tabular}{|p{0.95\columnwidth}|}
    \hline
\begin{Algoritmo}
\label{algo:GRAPHEM}
GraphEM algorithm
\begin{enumerate}
  \item[] \textbf{Inputs.} Prior parameters $\meanPrior_0$\;and  $\covPrior_0$; model parameters \cred{$\Q$}, $\{\H_k \}_{k=1}^K$, and $\{\R_k \}_{k=1}^K$; set of observations $\{\y_k \}_{k=1}^K$, and prior $p(\A)$. Precisions $(\varepsilon,\xi) >0$.  
\item[] \textbf{Initialization.} Set $\Ab^{(0)} \in \mathbb{R}^{N_x \times N_x}$.
\item[] \textbf{Recursive step.} For $i=0,1,\ldots$:   
 \begin{enumerate}
	 \item[(E step)] Run the Kalman filter and RTS smoother using transition matrix~$\A^{(i)}$. 
	\item[] Calculate $(\Psib,\Deltab,\Phi)$ using \eqref{eq:Sig}-\eqref{eq:C}-\eqref{eq:Phi}.
	\item[] Build function $\A \mapsto \mathcal{Q}(\A,\A^{(i)})$ using \eqref{eq:FuncQp}.
	\item[(M step)] Run Algorithm~\ref{algo:MS} with precision $\xi$ to solve
	\item[] $\A^{(i+1)} = \text{argmin}_{\A}\mathcal{Q}(\A,\A^{(i)})$.
 \end{enumerate}
\item[] If $\|\A^{(i+1)}-\A^{(i)}\|_F \leq \varepsilon \|\Ab^{(i)}\|_F$, \textbf{stop the recursion}.
 %
\item[] \textbf{{Output.}} State filtering/smoothing pdfs along with MAP estimate of the transition matrix.
\end{enumerate}
\end{Algoritmo}\\
        \hline
\end{tabular}
\end{table}

\subsection{\cblue{Explicit E-step}}
\label{ssec_Estep}

In this section, we derive the explicit E-step for the case of unknown $\A$. 
Let us first define the log-likelihood of the observations and states (that we recall, are not observed) that, due to the Markovian structure of the state space model in Eq.~\eqref{eq_model_state}, takes this form:
\begin{align}
\log p(\x_{0:K},\y_{1:K} | \Ab)  &= \log p(\x_0) + \sum_{k=1}^K \log p(\x_{k}|\x_{k-1}, \Ab) \nonumber \\
&+ \sum_{k=1}^K \log p(\y_{k}|\x_k).
\label{eq_log_joint_lk}
\end{align}
Following Section~\ref{sec:BackEM}, we must compute the expectation function $q(\Ab;\Ab^{(i)})$, given in \eqref{eq:funQ0}, i.e., the  log-likelihood of the observations and states integrated against the smoothing posterior $p(\x_{0:K},\y_{1:K} | \Ab)$, in such a way the states are marginalized and, therefore, the resulting function depends only on the model parameters. Function $\mathcal{Q}(\A;\A^{(i)})$ used in the Alg. \ref{algo:GRAPHEM} is then deduced easily from \eqref{eq:FuncQp}. 

{In the case of the LG-SSM in Eqs.~\eqref{eq_model_state}-\eqref{eq_model_obs}, {we now show that} there exists a closed-form expression for the integral \eqref{eq:funQ0}. Our approach uses the outputs of the RTS smoother and {generalizes} \cite[Theo. 12.4]{Sarkka}.  
{In a nutshell, we will demonstrate that} (a) the three log-quantities in Eq. \eqref{eq_log_joint_lk} are quadratic, and (b) the resulting integral \eqref{eq:funQ0} is tractable. Our proof lies in that the Kalman filter in Alg. \ref{alg_kf} provides an exact {form} $p(\x_k|\y_{1:k}) = \mathcal{N}(\x_k|\meanFilt_k,\covFilt_k)$, for every $k=1,\ldots,K$. The sequence of smoothing distributions (conditioned to the whole set of observations), can be also computed exactly by the RTS smoother in Alg. \ref{alg_rts}, yielding $p(\x_k|\y_{1:K}) = \mathcal{N}(\x_k|\meanSmooth_k,\covSmooth_k)$, for every $k=1,\ldots,K$.}  

Note that the computation of \eqref{eq:funQ0} requires the marginalization of the three terms in \eqref{eq_log_joint_lk}. However, since  M-step aims at minimizing $\mathcal{Q}(\A;\A^{(i)})$ w.r.t. $\A$,  only terms of \eqref{eq:funQ0} depending on variable $\A$ are in practice needed for the update, i.e., the second term of \eqref{eq_log_joint_lk}: 

{\footnotesize
\begin{align}
f& (\x_{1:K},\A)  \triangleq  \sum_{k=1}^K \log p(\x_{k}|\x_{k-1}, \Ab) \\
&= -\frac{{1}}{2}\sum_{k=1}^K \left(\left(\x_k - \A \x_{k-1} \right)^\top \cred{\Q}^{-1}\left(\x_k - \A \x_{k-1}  \right) + 
\log|{2\pi} \cred{\Q}|\right).
\label{eq_log_joint_lk_2}
\end{align}
}
Then, skipping the constant terms independent from $\A$, Eq.~\eqref{eq:funQ0} can be rewritten as
{\footnotesize
\begin{align} 
q(\A;& \Ab^{(i)}) = \int f(\x_{1:K},\A) p(\x_{0:K} | \y_{1:K},\Ab^{(i)}) d\x_{0:K} +  \ctA,\\
&= \int \left( -\frac{{1}}{2}\sum_{k=1}^K\left(\x_k - \A \x_{k-1} \right)^\top {\Q}^{-1}\left(\x_k - \A \x_{k-1}  \right) \right)\nonumber \\
&\quad \times p(\x_{0:K} | \y_{1:K},\Ab^{(i)}) d\x_{0:K} +  \ctA\\
&=  -\frac{{1}}{2}  \sum_{k=1}^K\int\left(\x_k - \A \x_{k-1} \right)^\top {\Q}^{-1}\left(\x_k - \A \x_{k-1}  \right) \nonumber\\
&\quad \times p(\x_{0:K} | \y_{1:K},\Ab^{(i)}) d\x_{0:K} +  \ctA.
\end{align}} 
\cblue{Then, we marginalize part of the variables to obtain}
{\footnotesize
\begin{align} 
q(\A;& \Ab^{(i)}) =  -\frac{{1}}{2}  \sum_{k=1}^K\int\left(\x_k - \A \x_{k-1} \right)^\top {\Q}^{-1}\left(\x_k - \A \x_{k-1}  \right) \nonumber\\
&\quad \times p(\x_{k:k-1} | \y_{1:K},\Ab^{(i)}) d\x_{k:k-1} +  \ctA\\
&=  -\frac{1}{2}  \sum_{k=1}^K\int\left(\x_k - \A \x_{k-1} \right)^\top {\Q}^{-1}\left(\x_k - \A \x_{k-1}  \right)  \nonumber\\
&\quad \times \mathcal{N}(\x_{k:k-1}|\meanSmooth_{k:k-1},\covSmooth_{k:k-1}) d\x_{k:k-1} +  \ctA,
\label{eq_I2}
\end{align}
}               
where $\mathcal{N}(\x_{k:k-1}|\meanSmooth_{k:k-1},\covSmooth_{k:k-1})$ denotes the joint smoothing distribution of two consecutive states $\x_{k:k-1} = [\x_{k} ; \x_{k-1}] \in \Real^{2N_x}$. The latter is Gaussian with mean 
\begin{align}
\meanSmooth_{k:k-1} = [\meanSmooth_{k} ; \meanSmooth_{k-1}],
\end{align}
and covariance
\begin{align}
{\covSmooth_{k:k-1} = [\covSmooth_{k} , \covSmooth_{k} \G_{k-1}^\top;\G_{k-1}\covSmooth_{k} , \covSmooth_{k-1}].}
\end{align}
The matrix  
$\G_{k} = \covPred_k \A^{(i)\top} \left(\A^{(i)}\covPred_k \A^{(i)\top} + {\Q} \right)$
follows from the derivation of the RTS smoother via manipulations of Gaussian pdfs (see for instance \cite[Theorem 8.2]{Sarkka}).  
Then, by defining $\widetilde{\Ab} = [\Id_{N_x} , - \Ab]$, Eq. \eqref{eq_I2} turns
{\footnotesize
\begin{align}
q(\A;\Ab^{(i)}) & = -\frac{{1}}{2} \sum_{k=1}^K\int\left( \widetilde{\Ab} \x_{k:k-1} \right)^\top {\Q}^{-1}\left(\widetilde{\Ab} \x_{k:k-1}   \right)  \nonumber\\
& \quad \times \mathcal{N}(\x_{k:k-1}|\meanSmooth_{k:k-1},\covSmooth_{k:k-1}) d\x_{k:k-1} +  \ctA\\
& = - \frac{{1}}{2} \sum_{k=1}^K \int \x_{k:k-1}^\top (\widetilde{\Ab}^\top {\Q}^{-1} \widetilde{\Ab}) \x_{k:k-1} \nonumber\\
& \quad \times \mathcal{N}(\x_{k:k-1}|\meanSmooth_{k:k-1},\covSmooth_{k:k-1}) d\x_{k:k-1} +  \ctA.
\label{eq:I2temp0}
\end{align}
}
We now apply equality \eqref{eqappA} in Appendix \ref{appendix_i1} (with {$\X \equiv \x_{k:k-1}$, $\mub \equiv 0$, $\Sigmab^{-1} \equiv \widetilde{\Ab}^\top {\Q}^{-1} \widetilde{\Ab}$, $\widetilde{\x} \equiv \meanSmooth_{k:k-1}$ and $\widetilde{\P} \equiv \covSmooth_{k:k-1}$}) to the integral term in \eqref{eq:I2temp0} (equality \eqref{eq:I2temp}(a)) and then the result \eqref{eqappb} in Appendix~\ref{appendix_i2} (equality \eqref{eq:I2temp}(b)):
{\footnotesize
\begin{align}
& \int  \x_{k:k-1}^\top \widetilde{\Ab}^\top {\Q}^{-1} \widetilde{\Ab} \x_{k:k-1} \mathcal{N}(\x_{k:k-1}|\meanSmooth_{k:k-1},\covSmooth_{k:k-1}) d\x_{k:k-1} 
\nonumber\\
& \overset{(a)}{=} 
\text{tr}\left(\widetilde{\Ab}^\top {\Q}^{-1} \widetilde{\Ab} (\covSmooth_{k:k-1} + \meanSmooth_{k:k-1} (\meanSmooth_{k:k-1})^\top)\right), \notag \\
& \overset{(b)}{=}  \text{tr}\left( {\Q}^{-1} (\covSmooth_k + \meanSmooth_{k} (\meanSmooth_{k})^\top - \Ab   (\Gb_{k-1} \covSmooth_k + \meanSmooth_{k-1} (\meanSmooth_{k})^\top )\nonumber\right. \\
& \left. {-} (\covSmooth_k \Gb_{k-1}^\top  + \meanSmooth_{k} (\meanSmooth_{k-1})^\top ) \Ab^\top 
+ \Ab ( \covSmooth_{k-1} + \meanSmooth_{k-1} (\meanSmooth_{k-1})^\top) \Ab^\top ) \right).
\label{eq:I2temp}
\end{align}
}
Finally, summing \eqref{eq:I2temp} for $k$ from $1$ to $K$, using the additivity property of the trace, and plugging the result into \eqref{eq:I2temp0}, yields 
{\small
{
\begin{equation}
q(\A;\Ab^{(i)}) = -\frac{1}{2} \text{tr}\left(\Q^{-1}(  
\Psib - \Deltab \Ab - \Ab \Deltab^\top + \Ab \Phib \Ab^\top) \right) +  \ctA,
\label{eq:g}
\end{equation}
}
}
with
\cred{
\begin{align}
\Psib & = \sum_{k=1}^K  \left(\covSmooth_k + \meanSmooth_{k} (\meanSmooth_{k})^\top\right), \label{eq:Sig}\\
\Deltab & =  \sum_{k=1}^K  \left(\covSmooth_k \Gb_{k-1}^\top  + \meanSmooth_{k} (\meanSmooth_{k-1})^\top\right), \label{eq:C} \\
\Phib & = \sum_{k=1}^K  \left(\covSmooth_{k-1} + \meanSmooth_{k-1} (\meanSmooth_{k-1})^\top\right).
\label{eq:Phi}
\end{align}
}
{The terms $(\Psib,\Deltab,\Phib)$ depend, in an implicit manner, of $\A^{(i)}$, that is the value of the transition matrix used when running the E-step (i.e., Kalman/RTS iterates). We omitted this dependency for the sake of readability.}

Then, using \eqref{eq:EMmaj}, we deduce that \eqref{eq:FuncQp}, where function $q$ given in \eqref{eq:g}, majorizes the MAP loss function $\loss_K$ in Eq.~\eqref{eq:phik} for every $\A \in \mathbb{R}^{N_x \times N_x}$.  {As a special case, when no prior is included and the noise covariance and observation matrices do not vary with $k$, we retrieve the result \cite[Theo.12.4]{Sarkka}.}  
 
\subsection{Computation in the M-step}
\label{sec_m_step}

{The M-step at iteration $i \in \mathbb{N}$ amounts to minimizing function $\mathcal{Q}(\A ;\A^{(i)})$ given in \eqref{eq:FuncQp}. Following the computations of the E-step, and particularly the result in \eqref{eq:g}, we can express this function in a generic form: 
\begin{equation}
(\forall \A \in \mathbb{R}^{N_x \times N_x}) \quad \mathcal{Q}(\A ;\A^{(i)}) = \sum_{m=1}^M f_m(\A),
\label{eq:FuncQsum}
\end{equation}
where 
\cred{
\begin{multline}
(\forall \A \in \mathbb{R}^{N_x \times N_x}) \\
f_1(\A) = \frac{1}{2}  \text{tr} \left(\Q^{-1}(\Psib - \Deltab \A^\top - \A \Deltab^\top + \A \Phib \A^\top) \right),\label{eq_f1}\\
\end{multline}
}
and $\sum_{m=2}^M f_m(\A) = \loss_0(\A)$ is the regularization term. We recall that, for every $\A \in \mathbb{R}^{N_x \times N_x}$, $f_1(\A) = -q(\A,\A^{(i)}) + \ctA$ (i.e., up to a constant independent from $\A$) while $\loss_0(\A) = -\log(p(\A))$. The assumed sum structure for $\loss_0$ allows us to account for factorizing priors.}

Function $f_1$ in \eqref{eq_f1} is quadratic and convex on $\mathbb{R}^{N_x \times N_x}$. We furthermore assume that each $\{f_m\}_{m=2}^M$ involved in the regularization term is proper (i.e., with non empty domain), convex, and lower semi-continuous on $\mathbb{R}^{N_x \times N_x}$, and such that the set of minimizers of \eqref{eq:FuncQsum} is non-empty. Function~\eqref{eq:FuncQsum} consequently reads as {a sum of a quadratic function, and convex possibly non smooth terms.} This paves the way for the application of primal-dual proximal approach for its minimization. Primal-dual proximal splitting (PDPS) algorithms \cite{komodakis2015playing} rely on the fundamental tool called the proximity operator, whose definition is stated as follows. For a proper, lower semi-continuous and convex function $f: \mathbb{R}^{N_x \times N_x} \mapsto (-\infty,+\infty]$, the proximity operator\footnote{See also \url{http://proximity-operator.net/}} of $f$ at $\widetilde{\A} \in \mathbb{R}^{N_x \times N_x}$ is defined as~\cite{Combettes2011}
\begin{equation}
    \text{prox}_f(\widetilde{\A}) = \text{argmin}_\A \left( f(\A) + \frac{1}{2}\| \A - \widetilde{\A}\|^2_F \right).
\end{equation} 
Given this tool, a generic PDPS method can iteratively minimize \eqref{eq:FuncQsum} by processing sequentially the terms $\{f_m\}_{m=1}^M$, either through their gradient or their proximity operator. The convergence of the sequence to a minimizer of \eqref{eq:FuncQsum} is then guaranteed, under specific rules on the algorithm hyperparameters (e.g., the stepsize). {A large number of algorithms can be built from this generic strategy, with different practical efficiency, depending on several factors such as the order of the updates, the way to process linear operators, the stepsize rules, the use or not of randomized block updates, etc. \cite{raguet2013generalized,chambolle2011first,Combettes2021}.} 

{
On the one hand, following the comparative analysis from \cite{Glaudin2019,Briceno2021}, we will prefer an algorithm that activates each terms via their proximity operator. Function $f_1$ is quadratic and with a close form for its proximity operator. Indeed, for every $\vartheta>0$, for every $\A \in \mathbb{R}^{N_x \times N_x}$,}
{\small
\begin{equation}
    \text{prox}_{\vartheta f_1}(\A) =  \operatorname{lyapunov}\left(\vartheta \Q^{-1}, \Phib^{-1}, {\A}\Phib^{-1} + \vartheta \Q^{-1} \Deltab \Phib^{-1}\right),
                \label{eq:proxf1}
\end{equation}
}
where $A = \operatorname{lyapunov}(X,Y,Z)$ provides the solution to the Lyapunov equation $XA + AY = Z$ \cite{Combettes2011}. If $\Q = \sigma_{\Q}^2 \rm{\Id}_{N_x}$ for some $\sigma_Q>0$, then \eqref{eq:proxf1} simplifies into
\begin{equation}
{\small
    \text{prox}_{\vartheta f_1}(\A) =  \left(\frac{\vartheta}{\sigma_{\Q}^2} \Deltab + {\A} \right)  \left(\frac{\vartheta}{\sigma_{\Q}^2}\Phib +  \rm{\Id}_{N_x}  \right)^{-1}.
                \label{eq:proxf1s}
                }
\end{equation}
{On the other hand, it might be beneficial to impose the M-step update to satisfy certain structural properties, such as sparsity, regardless the precision level of its implementation.  
For the two aforementioned reasons, we opt for the monotone+skew (MS) algorithm from \cite{Briceno2011}, described in Algorithm~\ref{algo:MS}.}. {More precisely, we assume without loss of generality that $f_M$ is our sparsity-enhancing term. We then propose an implementation of the approach of \cite{Briceno2011} where we particularize $f_M$ while all the remaining terms $\{f_m\}_{1 \leq m \leq M-1}$ are processed in a consensus-based manner \cite{PesquetRepetti2015,komodakis2015playing}.} 
In this way, the output of Alg.~\ref{algo:MS} inherits the structure of the proximity operator of $f_M$. For instance, if $f_M$ is the $\ell_1$ norm then the output of Alg.~\ref{algo:MS} is sparse by construction~\cite{BeckISTA,Daubechies}, whatever the value of the precision parameter $\xi$. Algorithm~\ref{algo:MS} has two other parameters besides the precision level, namely the stepsizes $(\gamma,\lambda)$ whose choice is dictated by the convergence analysis. Under the range settings of Alg.~\ref{algo:MS}, we can establish the following Proposition~\ref{ref:theoMS}. 

{                                  
\begin{Proposicion}
\label{ref:theoMS}
Assume that, for every $m \in \{1,\ldots,M\}$, function $f_m$ is convex, proper, and lower semicontinuous on $\mathbb{R}^{N_x \times N_x}$. Then, the sequences $\{\A_n^M\}_{n \in \mathbb{N}}$ and $\{\Z_n^M\}_{n \in \mathbb{N}}$ converge to a minimizer of \eqref{eq:FuncQsum}. 
\end{Proposicion}
\begin{proof}
The proof relies on applying the consensus-based splitting from \cite[Sec. III]{komodakis2015playing} to $\sum_{m=1}^{M-1} f_m$. Let us introduce $\mathbf{L} = [\Id_{N_x},\ldots,\Id_{N_x}]^\top \in \mathbb{R}^{(M-1) N_x \times N_x}$  and $g: \mathbb{R}^{(M-1) N_x \times N_x} \to (-\infty,+\infty]$ such that, for every $\V = [\V_1^\top,\ldots,\V_{M-1}^\top]^\top\in \mathbb{R}^{(M-1) N_x \times N_x}$, $g(\V) = \sum_{m=1}^{M-1} f_m (\V_m)$. Then, minimizing \eqref{eq:FuncQsum} is equivalent to minimize
\begin{equation}
(\forall \A \in \mathbb{R}^{N_x \times N_x}) \quad g(\mathbf{L} \A) + f_M(\A).
\end{equation}
By construction, $\|\mathbf{L}\|_2 = M-1$. Moreover, for every $\V = [\V_1^\top,\ldots,\V_{M-1}^\top]^\top\in \mathbb{R}^{(M-1) N_x \times N_x}$, $\mathbf{L}^\top \V = \sum_{m=1}^{M-1} \V_m$ and $\text{prox}_g(\V) = [\text{prox}_{f_1}(\V_1)^\top,\ldots,\text{prox}_{g_{M-1}}(\V_{M-1})^\top]^\top$. Then, the proposed Alg. \ref{algo:MS} identifies with \cite[Eq. (4.8)]{Briceno2011} and applying \cite[Prop. 4.2]{Briceno2011} concludes the proof. 
\end{proof}
}
Typical choices for setting the hyper-parameters in Alg. \ref{algo:MS}, satisfying the range assumptions and adopted in our experiments, are
\begin{equation}
\lambda = \frac{0.9}{M}, \quad \gamma = \frac{1-\lambda}{M-1}.
\end{equation}
In practice, the algorithm is stopped as soon as \eqref{eq:FuncQsum} stabilizes. Note that a warm start initialization strategy is employed in Alg. \ref{algo:MS}. The so-called dual variables $\{\V_0^m\}_{m=1}^M$ are set to the previous estimate of the state matrix, that is $\A^{(i)}$. This was observed to yield considerable reduction of required iterations to reach our stopping criterion, when compared to a cold start (setting initial dual variables to zero, for instance). {Our implementation of MS processes separately function $f_M$, and the other terms $(f_m)_{1 \leq m \leq M-1}$. In particular, the elements of the converging sequence $\{\A_n^M\}_{n \in \mathbb{N}}$ of Alg. \ref{algo:MS} are outputs of proximity operator for $f_M$, and thus are sparse for suitable choice of this regularization term. This feature is not present in most standard implementations of primal-dual proximal splitting techniques.}

\begin{table}[!t]
\vspace{4mm}
    \centering
    \begin{tabular}{|p{0.95\columnwidth}|}
    \hline
\begin{Algoritmo}
\label{algo:MS}
MS algorithm for GraphEM M-step
\begin{enumerate}
  \item[] \textbf{Inputs.} $\Ab^{(i)}, \Psib, \Deltab, \Phi\cred{, \Q}$, and prior $p(\A)$. Precision $\xi>0$.
\item \textbf{Setting.} Set stepsizes $\lambda \in (0,1/M)$, $\gamma \in [\lambda,(1-\lambda)/(M-1)]$.
  \item \textbf{Initialization.} For every $m \in \{1,\ldots,M\}$, $\V_0^m = \Ab^{(i)}$.
\item \textbf{Recursive step.} For $n=1,2,\ldots$:  
\begin{equation}
\begin{array}{l}
\W_n^m = \V_n^m + \gamma \V_n^M \, (\forall m \in \{ 1,\ldots,M-1\})\\
\W_n^M = \V_n^M - \gamma \sum_{m=1}^{M-1} \V_n^m\\
\A_n^m = \W_n^m - \gamma \operatorname{prox}_{f_m/\gamma}(\W_n^m) \, (\forall m \in \{ 1,\ldots,M-1\})\\
\A_n^M = \operatorname{prox}_{\gamma f_M}(\W_n^M)\\
\Z_n^m  = \A_n^m + \gamma  \A_n^M \, (\forall m \in \{ 1,\ldots,M-1\})\\
\Z_n^M  = \A_n^M - \gamma \sum_{m=1}^{M-1} \A_n^m\\
\V_{n+1}^m = \V_n^m - \W_n^m + \Z_n^m \, (\forall m \in \{ 1,\ldots,M\}).
 \end{array}
\label{eq_MS}
\end{equation}
\item[] If $|\mathcal{Q}(\A_n^M,\A^{(i)})-\mathcal{Q}(\A_{n-1}^M,\A^{(i)})| \leq \xi$, \textbf{stop the recursion}.
 %
\item[] \textbf{{Output.}} Transition matrix update, $\A^{(i+1)} = \A_n^M$.
\end{enumerate}
\end{Algoritmo}\\
        \hline
\end{tabular}
\end{table}

\subsection{Choice of the prior}
\label{sec_prior}

Let us now explicit choices for the prior $p(\A)$, and thus for the regularization function $\loss_0$, that are encompassed by our study. We will focus on a hybrid form for the regularization function, such that
\begin{equation}
(\forall \A \in \mathbb{R}^{N_x \times N_x}) \quad \loss_0(\A) = \sum_{m=2}^M f_m(\A).
\end{equation}
For the sake of readability, we denote $f$ a possible regularization function, keeping in mind that $\loss_0$ might combine various terms, to promote various properties in the state matrix $\A$. All these terms would be then processed in a parallel manner in Alg. \ref{algo:MS}, through their proximity operator. 
 
Let us first discuss a particularly useful choice for $f$ covered by our study. An important matter is to make sure that the LG-SSM resulting from the parameter identification phase (here, the EM procedure) presents good structural properties. In particular, one may require that the first order auto-regressive model inherent to the {state process in} LG-SSM is stable, in order to avoid any numerical divergence for large values of $K$. The stability is directly related to the spectral properties of matrix $\A$ in \eqref{eq_model_state}. {As a result, a sufficient condition for the LG-SSM to be stable (i.e., not diverging with $K \to \infty$) is to be parameterized by an $\A$ parameter with singular values less than one~\cite{bird2019customizing,shumway2000time}. This condition can be incorporated within our framework}, by defining
\begin{align}
(\forall \A \in \mathbb{R}^{N_x \times N_x}) \quad f(\A) & = \begin{cases}
0 & \text{if } \A \in \mathcal{S}\\
+\infty & \text{elsewhere} \label{eq:indicator}
\end{cases}\\
& \triangleq \iota_{\mathcal{S}}(\A).
\end{align}
{
Hereabove, $\mathcal{S} \subset \mathbb{R}^{N_x \times N_x}$ is related to the stability condition on the SSM:
\begin{equation}
\mathcal{S} = \{ \A \in \mathbb{R}^{N_x \times N_x} | \| \A \|_2 \leq \delta < 1\}, \label{eq:setS}
\end{equation}
for $\delta \in (0,1)$ (typically close to one)}. The proximity operator of \eqref{eq:indicator} is simply the projection onto $\mathcal{S}$. Such projection has a closed form \cite{bauschke2017convex}, that we explicit in Table~\ref{tab:constraint}. We also provide in Table~\ref{tab:constraint} two other meaningful examples for $\mathcal{S}$, along with the expression for the associated projection. {In particular, the range constraint can be used to impose the sign of some entries of $\A$, i.e., to impose the arrows direction in the estimated graph (under our interpretation).}

We now continue our discussion by presenting another family of possible penalty terms in $\loss_0$. For each presented example for $f$, we provide the expression for $\operatorname{prox}_{\vartheta f}$ where $\vartheta$ is a positive scaling parameter. This is the aim of Table~\ref{tab:prior}. We focus on two particular choices presented in the table, namely the Laplace and block-Laplace priors. Both choices enhance sparsity of matrix $\A$, with the latter being a generalization of the former. The introduction of sparsity promoting prior is in general desirable when doing parameter identification, and is key under our novel approach. As any regularization, it aims at reducing over-fitting problems that could arise for low values of $K$ and thus increases the generalization capacity of the model. Even more, it promotes matrices $\A$ with few non-zero entries, which highly helps for the interpretability of the resulting SSM. Each non-zero entry can actually be understood as a statistical dependence (correlation in this case), between two state dimensions in two consecutive time steps. One can thus interpret $\A$ as the adjacency matrix of a directed graph (since the entries of $\A$ are signed) mapping the entries of the hidden state vector from time $k-1$ to those of time $k$. The GraphEM approach proposed in this work aims at recovering this graph, and if possible, promoting an interpretable structure. This is done by incorporating a prior of sparsity on $\A$, thanks to appropriate choice for $f$. An immediate idea would be to define $f$ as the $\ell_0$ norm of $\A$, that counts the number of non-zero entries of the matrix. However, this function is non-convex, non continuous, and it is associated to a improper law $p(\A)$, which is undesirable. Instead, one prefers to choose for $p(\A)$, the proper, log-concave Laplace distribution, leading to the so-called Lasso regularization \cite{BachSparse} reading $f(\A) = \kappa \ell_1(\A)$ with $\kappa>0$ a regularization weight. The larger $\kappa$, the stronger sparsity of $\A$, with the extreme case of a null $\A$ for sufficiently large $\kappa$. The $\ell_1$ norm has been used in numerous works of signal processing and machine learning~\cite{Tibshirani,Chaux_2007}, including graph signal processing \cite{Friedman07,Benfenati18}. It has a simple closed form proximity operator, namely the soft thresholding operator \cite{Daubechies}, that we recall in Table~\ref{tab:prior}. 

In certain scenarios, one might have some prior knowledge about some structured sparsity in $\A$. Otherwise stated, one might want to cancel (or not) some blocks of $\A$ in a simultaneous manner, because the entries of these blocks are connected. For instance, they could correspond to real/imaginary part of the same complex quantity (see example in the experimental section). This paves the way for using a more sophisticated prior, where the Laplace distribution is now promoted on each block of $\A$. More formally, let $B \geq 1$ a divisor of $N_x^2$, defining the number of blocks. Each $\A \in \mathbb{R}^{N_x \times N_x}$ can be rewritten equivalently as a set of $B$ vectors $(\mathbf{a}(b))_{1 \leq b \leq B}$ of size $N_x^2/B$. The block-Laplace prior amounts to computing the $\ell_2$ norms of each of these vectors, and then summing the $B$ resulting values, to obtain $f(\A)$. This can also be rewritten $f(\A) = \kappa \ell_{2,1}(\A)$, by using the mixed norm notation from \cite{KOWALSKI2009303}, and introducing the regularization weight $\kappa>0$. Mixed norms have been widely used in machine learning  under various combinations \cite{BachSparse}. For our particular choice, the proximity operator remains simple, and is provided in Table~\ref{tab:prior}. It is worth noticing that both proximity operators for Laplace and block-Laplace involve a threshold of the entries of their input. We list two other examples of priors in the table, the former being the result of a Gaussian prior distribution, while the latter combines Laplace and Gaussian, and is also known as \emph{Lasso with elastic-net} \cite{zou2005regularization}. 
  
\begin{table*}
\caption{Example of priors, expressions for the resulting regularization and its proximity operator with scale parameter $\vartheta>0$.}
\centering
\renewcommand{\arraystretch}{2}
\begin{tabular}{|c|c|c|}
\hline
Prior & $f(\A)$ & $\operatorname{prox}_{\vartheta f}(\A)$\\
\hline
\hline
Laplace & $ \|\A \|_1 = \sum_{n = 1}^{Nx} \sum_{\ell=1}^{N_x} |A(n,\ell) |$ & $\left(\text{sign}(A(n,\ell) ) \times\max(0,|A(n,\ell) | - \vartheta)
    \right)_{1 \leq n,\ell\leq N_x}$\\
\hline
Block-Laplace & $ \|\A \|_{2,1} = \sum_{b=1}^{B} \|\mathbf{a}(b)\|_2 $ & $\left( (1 - \frac{\vartheta}{\max(\|\mathbf{a}(b)\|_2,\vartheta)}) \, \mathbf{a}(b)
    \right)_{1 \leq b \leq B}$ \\
\hline
Gaussian & $\frac{1}{2} \|\A \|_F^2 = \frac{1}{2}\sum_{n = 1}^{Nx} \sum_{\ell=1}^{N_x} (A(n,\ell) )^2$ & $\left(\frac{A(n,\ell) }{1 + \vartheta}
    \right)_{1 \leq n,\ell \leq N_x}$ \\
\hline
Laplace + Gaussian & $ \|\A \|_1 + \frac{1}{2} \| \A \|_F^2$ & $\left(\text{sign}\left(\frac{A(n,\ell) }{1 + \vartheta}\right) \times\max\left(0, \left|\frac{A(n,\ell) }{1 + \vartheta} \right| - \frac{\vartheta}{1 + \vartheta}\right)
    \right)_{1 \leq n,\ell \leq N_x}$\\
\hline
\end{tabular}
\label{tab:prior}
\end{table*}

\begin{table*}
\caption{Example of convex constrained sets and associated projection operators. $\delta>0$ and $a_{\min}\leq a_{\max}$ are hyper-parameters. We use the singular value decomposition $\A = \mathbf{U}^\top \text{Diag}(\mathbf{s}) \mathbf{V}$.}
\centering
\renewcommand{\arraystretch}{2}
\begin{tabular}{|c|c|c|}
\hline
Constraint & $\mathcal{S}$ & $\mathrm{Proj}_{\mathcal{S}}(\A)$\\
\hline
\hline
Bounded spectrum & $ \|\A\|_2 \leq \delta$ & $\mathbf{U}^\top \text{Diag} \left(\left(\text{sign}(s_n)\min(|s_n|,\delta)  
    \right)_{1 \leq n \leq N_x}\right) \mathbf{V}$\\
\hline
Range & $(\forall (n,\ell) \in \{1,\ldots,N_x\}^2) \; A(n,\ell) \in [a_{\min},a_{\max}]$ &  $\left( \min(\max(a_{\min},A(n,\ell) ),a_{\max}) 
    \right)_{1 \leq n,\ell \leq N_x}$\\
 \hline
Bounded energy & $\|\A\|_F \leq \delta$ & $\left( (1 - \frac{\delta}{\max(\|\A\|_F,\delta)}) \, A(n,\ell) 
    \right)_{1 \leq n,\ell \leq N_x}$\\
\hline
\end{tabular}
\label{tab:constraint}
\end{table*}

\subsection{Convergence result}
\label{sec:conver}

\cblue{We now show the convergence of GraphEM as in  Algorithm~\ref{algo:GRAPHEM}. We refer to \cite{bauschke2017convex} for definitions of functional analysis.}
\cblue{
\begin{Teorema}
\label{ref:graphEM}
Assume that the MAP loss function \eqref{eq:phik} is coercive on $\mathbb{R}^{N_x \times N_x}$ and that the prior term $\mathcal{L}_0$ is proper, convex, and lower semicontinuous on $\mathbb{R}^{N_x \times N_x}$. We furthermore assume that the relative interior of the domain of $\mathcal{L}_0$ contains the level set $\mathcal{E} = \{\Ab \in \mathbb{R}^{N_x \times N_x} | \mathcal{L}_K(\Ab) \leq \mathcal{L}_K(\Ab^{(0)})\}$. If the M-step in GraphEM is solved exactly i.e., for every $i \in \mathbb{N}$,
\begin{equation}
\A^{(i+1)} = \text{argmin}_{\A \in \mathbb{R}^{N_x \times N_x}} \mathcal{Q}(\A;\A^{(i)}), \label{eq:graphEMex}
\end{equation}
with $\Ab^{(0)} \in \mathbb{R}^{N_x \times N_x}$, then the following statements hold:
\begin{itemize}
\item[(i)] The sequence $\left(\mathcal{L}_K(\A^{(i)})\right)_{i \in \mathbb{N}}$ is a decreasing sequence converging to a finite limite $\mathcal{L}^*$.
\item[(ii)] The sequence of iterates $\left(\A^{(i)}\right)_{i \in \mathbb{N}}$ has a cluster point (i.e., one can extracts a converging subsequence).
\item[(iii)] Let $\A^*$ a cluster point (i.e., the limit of a converging subsequence) of $\left(\A^{(i)}\right)_{i \in \mathbb{N}}$ . Then, $\mathcal{L}_K(\A^*) = \mathcal{L}^*$, and $\A^*$ is a critical point of $\mathcal{L}_K$, i.e., $\nabla \mathcal{L}_{1:K}(\A^*) \in \partial \mathcal{L}_0(\A^*)$. 
\end{itemize}
\end{Teorema}
\begin{proof} 
Our proof consists in showing that the conditions of \cite[Th.~5]{Benfenati18} are met. First, let us remark that the GraphEM exact formulation \eqref{eq:graphEMex} is well defined, since for every $\Ab$, $\mathcal{Q}(\A;\A^{(i)})$ is a coercive lower-semicontinuous function. It indeed majorizes $\mathcal{L}_K$ which is coercive by assumption.  Moreover, according to \cite{Gupta1974}, the likelihood function $\mathcal{L}_{1:K}$ is continuously differentiable on $\mathbb{R}^{N_x \times N_x}$. In particular, it is continuously differentiable on the level set $\mathcal{E}$. The rest of the proof follows using the same arguments as in \cite[Th.~5]{Benfenati18}, and the subdifferential calculation from \cite[Corollary 16.48(ii)]{bauschke2017convex}.
\end{proof}
}

\cblue{First, it should be noted that this result focuses on the exact form of Alg.~\ref{algo:GRAPHEM}, when the M-step is assumed to be solved exactly. Extending Theorem~\ref{ref:graphEM}(i)-(ii) to the case of an inexact resolution of the M-step would be straightforward, but it is not the case for Theorem~\ref{ref:graphEM}(iii). According to our Proposition~\ref{ref:theoMS}, the sequence produced by our M-step inner solver in Alg. \ref{algo:MS} converges to an optimal value. In practice, we did not observe any numerical instabilities of the algorithm as soon as a sufficient precision was imposed in the M-step resolution. Second, let us notice that our assumptions on the penalty term $\mathcal{L}_0$ are compliant with those made in Section~\ref{sec_m_step} and with the examples discussed in Section~\ref{sec_prior}. The case of a null penalty (i.e., $\mathcal{L}_0 \equiv 0$) is also covered by our assumptions. In such case, Alg.~\ref{algo:GRAPHEM} becomes equivalent to the MLEM algorithm from \cite{shumway1982approach}. Up to our knowledge, our convergence result is new even for this simple setting. Finally, we must emphasize that, due to the intricate form of the likelihood function and of the presence of a possibly non-differentiable penalty term, it appeared not possible to apply the standard convergence analysis for EM methods from \cite{Wu}.}

\section{\cblue{Numerical simulations}}
\label{sec:experiment}

\subsection{Synthetic data}
\label{sec:synthetic}

We start our experimental section by illustrating the performance of GraphEM in a controlled scenario involving synthetic data. Time series $\{\y_k,\x_k\}_{k=1}^K$ are simulated using  \eqref{eq_model_state}-\eqref{eq_model_obs}, with settings $K = 10^3$, \cred{$\Q = \sigma_\Q^2 {\Id}_{N_x}$}, $\R_k = \sigma_\R^2 {\Id}_{N_y}$ for every $k \in \{1,\ldots,K\}$, $\P_0 = \sigma_\P^2 {\Id}_{N_x}$ with $(\sigma_\Q,\sigma_\R,\sigma_\P)$ some predefined values. We consider the scenario where $\H_k = {\Id}_{N_x}$ for every $k \in \{1,\ldots,K\}$, that is {there is a one-to-one correspondence between states and observations, and thus $N_x = N_y$}. This choice presents the advantage of avoiding any identifiability issues that may arise from an ill-conditioned observation matrix, and thus to fully focus on the graph inference problem, i.e., the estimation of matrix $\A$. Since we are dealing with synthetic data, the ground truth matrix $\A$ can be predefined. In our experiments, we rely on block-diagonal matrices $\A$, made of $J$ blocks with dimensions $\{B_j\}_{1 \leq j \leq J}$, so that $N_y = \sum_{j=1}^J B_j$. 
The diagonal blocks of $\A$ are randomly set as matrices of auto-regressive processes of order one, AR(1), satisfying the stability assumption (i.e., spectral norm less than one). 
This procedure leads to the construction of four datasets summarized in Table~\ref{tab:data}. Having the groundtruth available allows us to rely on quality assessment metrics of the estimated $\A$. Here, we retain the relative mean square error (RMSE) \cblue{in the estimation of the transition matrix}, as well as the precision, recall, specificity, accuracy, and F1 score for detecting the non-zero entries of $\A$ (that is, the graph edges positions). A threshold value of $10^{-10}$ on the absolute entries of matrix $\A$ is used for the detection hypothesis.

\begin{table*}[t]
\caption{Results for GraphEM, StableEM, OracleEM, MLEM, PGC and CGC, along with averaged computing times.}
\centering
{\footnotesize
\cblue{
\begin{tabular}{|c|c||c|c|c|c|c|c||c|}
\cline{2-9}
\multicolumn{1}{c|}{ } & method & RMSE & accur. & prec. & recall & spec. & F1 & Time (s.)\\
\hline
\multirow{4}{*}{A} & GraphEM & $\bf{0.081789}$ & $\bf{0.90988}$ & $0.999$ & $0.73037$ & $0.99963$ & $\bf{0.84361}$ & $2.3063$\\
  & StableEM & $0.1405$ & $0.3333$ & $0.3333$ & $\bf{1}$ & $0$ & $0.5$& $2.3506$\\
  & MLEM & $0.148$ & $0.3333$ & $0.3333$ & $\bf{1}$ & $0$ & $0.5$& $1.6059$\\
  & PGC& - & $0.8765$ & $0.9474$ & $0.6667$ & $0.9815$& $0.7826$& $0.1312$\\ 
  & CGC & - & $0.8765$ & $\bf{1}$ & $0.6293$ & $\bf{1}$& $0.7727$& $0.1366$\\ 
	\cline{2-9}
	  & OracleEM & $0.0879$ & $1$ & $1$ & $1$ & $1$ & $1$& $0.9572$\\
\hline
\hline
\multirow{4}{*}{B} & GraphEM & $\bf{0.080687}$ & $\bf{0.90691}$ & $\bf{1}$ & $0.72074$ & $\bf{1}$ & $\bf{0.83753}$& $2.1448$\\
  & StableEM & $0.15042$ & $0.3333$ & $0.3333$ & $\bf{1}$ & $0$ & $0.5$& $3.0263$\\
& MLEM & $0.15203$ & $0.3333$ & $0.3333$ & $\bf{1}$ & $0$ & $0.5$& $1.5291$\\
  & PGC & - & $0.8889$ & $\bf{1}$ & $0.6667$ & $\bf{1}$& $0.8$& $0.0606$\\ 
  & CGC & - &$0.8889$ & $\bf{1}$ & $0.6667$ & $\bf{1}$& $0.8$& $0.0631$\\ 
	\cline{2-9}
	  & OracleEM & $0.076122$ & $1$ & $1 $ & $1 $ & $1 $ & $ 1$& $1.1179$\\
\hline
\hline
\multirow{4}{*}{C} & GraphEM & $\bf{0.12624}$ & $\bf{0.91695}$ & $0.97392$ & $0.70676$ & $0.99298$ & $\bf{0.81878}$& $5.0027$\\
  & StableEM & $0.23253$ & $0.2656$ & $0.2656$ & $\bf{1}$ & $0$ & $0.4198$& $5.6791$\\
& MLEM & $0.2448$ & $0.2656$ & $0.2656$ & $\bf{1}$ & $0$ & $0.4198$& $5.6557$\\
  & PGC& - & $0.9023$ & $\bf{0.9778}$ & $0.6471$ & $\bf{0.9949}$& $0.7788$& $0.4095$\\ 
  & CGC& - & $0.8555$ & $0.9697$ & $0.4706$ & $\bf{0.9949}$& $0.6337$& $0.4175$\\ 
	\cline{2-9}
	  & OracleEM & $0.1214$ & $1$ & $1 $ & $1 $ & $1 $ & $ 1$& $2.5504$\\
\hline
\hline
\multirow{4}{*}{D} & GraphEM& $\bf{0.12347}$ & $\bf{0.91648}$ & $\bf{0.98866}$ & $0.69382$ & $\bf{0.99702}$ & $\bf{0.81514}$& $2.9988$\\
  & StableEM & $0.22897$ & $0.2656$ & $0.2656$ & $\bf{1}$ & $0$ & $0.4198$& $4.9561$ \\
  & MLEM& $0.2416$ & $0.2656$ & $0.2656$ & $\bf{1}$ & $0$ & $0.4198$& $2.5501$\\
  & PGC & - & $0.8906$ & $0.9$ & $0.6618$ & $0.9734$& $0.7627$& $0.2881$\\ 
    & CGC& - & $0.8477$ & $0.9394$ & $0.4559$ & $0.9894$& $0.6139$& $0.2948$\\ 
	\cline{2-9}
		  & OracleEM & $0.11925$ & $1$ & $1 $ & $1 $ & $1 $ & $ 1$& $2.2367$\\
\hline
\end{tabular}
}
}
\label{tab:results}
\end{table*}

For each dataset, we ran GraphEM algorithm using a stable AR(1) matrix as initial estimate, and precision parameters $(\varepsilon,\xi) = (10^{-3},10^{-4})$. The regularization $\loss_0 = f_2 + f_3$ with $f_2 = \iota_{\mathcal{S}}$ indicator function of the stable matrix set \eqref{eq:setS} with $\delta = 0.99$, and $f_3 = \kappa \ell_1$ with weight parameter $\kappa>0$. Such choice satisfies the required assumptions for the convergence of the MS algorithm for the M-step. Parameter $\kappa$ is set empirically through a rough grid search maximizing the accuracy score. As for comparison, we also provide the results obtained when (i) no regularization is employed, thus leading to the ML estimator (MLEM), (ii) MLEM is modified so as to account for an oracle knowledge of the position of zero entries of $\A$ (OracleEM), and (iii) only stability constraint is imposed (StableEM). In each of these cases, a similar EM-based procedure than GraphEM is used, with simplified computations for the M-step. The results from OracleEM are separated from the others, as it requires the ground truth knowledge of the graph support, not available in practical situations. In addition to these EM-based methods, we provide comparisons with two Granger-causality approaches \cite{Bressler2011} for graphical modeling, namely pairwise Granger Causality (PGC) and conditional Granger Causality (CGC). Both methods provide a binary information about the identification (or not) of an edge in the graph, by relying on conditional dependency analysis. PGC explores the $N_x(N_x -1)$ possible dependencies among two nodes, at each time independently from the rest. CGC additionally accounts, for each pair of nodes, for the information of the other $N_x-2$ signals, in order to evaluate whether one node brings information to the other while the rest of signals are observed.
As PGC and CGC do not provide a weighted graph estimation, no RMSE score is computed in those case.

\begin{table}[h]
\caption{Description of datasets}
\centering
{\footnotesize
\begin{tabular}{|c||c|c|}
\hline
Dataset & $(B_j)_{1 \leq j \leq J}$ & $(\sigma_\Q,\sigma_\R,\sigma_\P)$\\
\hline\hline
A & $(3,3,3)$ & $(10^{-1},10^{-1},10^{-4})$\\
\hline
B & $(3,3,3)$ & $(1,1,10^{-4})$\\
\hline
C & $(3,5,5,3)$ & $(10^{-1},10^{-1},10^{-4})$\\
\hline
D & $(3,5,5,3)$ & $(1,1,10^{-4})$\\
\hline
\end{tabular}
}
\label{tab:data}
\end{table}
 
The results, averaged on $50$ realizations, are presented in Table~\ref{tab:results}. Nor MLEM neither StableEM promote sparsity in the graph which explains their poor results in terms of edge detection. Still, StableEM presents a slightly better RMSE score, showing the advantage of integrating the stability constraint as a prior during the estimation procedure. GraphEM provides very good RMSE score on all examples. It is remarkable that these scores are comparable, and sometimes even better, than those obtained with OracleEM. This shows that our construction for the regularization function, gathering both sparsity and stability terms, is well suited to reach a high quality estimate for the state matrix. Moreover, the retained $\ell_1$ penalty does not appear here to yield any bias in the estimated graph weights, as it can be sometimes noticed in Lasso regression \cite{Tibshirani}. This can be probably explained by the proposed combination of an $\ell_1$ term and the stability spectral constraint. Regarding the graph structure, we can observe that GraphEM has also better detection scores, when compared with both PGC and CGC. \cblue{We observe that GraphEM is consistently superior in accuracy and F1. These metrics are relevant since they take into account both the true positive and negative connections. Both StableEM and MLEM present a recall metric equal to one, which corresponds to an estimate of the transition matrix without any null entries. An opposite effect is observed by PGC and CGC in specificity, since both methods can over-estimate the amount of zeros (no connections), particularly the latter \cite{Luengo19}. We remind } that OracleEM should not be compared within this metric, since it has access to the edges position and thus has perfect edge detection scores. \cblue{We also provide in the last column of Table~\ref{tab:results}, the averaged computing times over $50$ realizations for each methods, for Matlab 2021a codes running on a 11th Gen Intel(R) Core(TM) i7-1185G7 3.00GHz with 32 Go RAM. PGC and CGC require the lowest computing times. These two methods are based on simpler auto-regressive processes (without latent sates), which explains their poorer performance w.r.t. GraphEM. The other methods share computing times with similar order of magnitude. OracleEM is the fastest method among the EM-based ones, simply because it works in the favorable setting when the edge positions of the graph are assumed to be known, thus reducing the size of the search space. GraphEM is very competitive, compared to its non-regularized counterpart MLEM, thanks to the proposed efficient proximal splitting M-step resolution, while reaching better quantitative results than MLEM by far. Interestingly, for a given dataset size, one can notice a trend of a lower computational times when solving the inference problem for an higher noise level (see dataset A vs B, dataset C vs D), whatever the algorithm employed. This might be related to the peaky likelihood phenomenon described in \cite{DelMoral}, namely the larger noise variance, the easier it is to explore the posterior.
}

We also display an example of graph reconstruction for dataset~C in Fig.~\ref{fig:datasetC}, illustrating the ability of GraphEM to recover the graph structure and signed weights. Finally, we show on Fig.~\ref{fig:datasetAplots} a comparison between MLEM and GraphEM, in terms of evolution of the loss function~\eqref{eq:phik} and the RMSE score, along the iterations of both algorithms. One can notice that both methods reach convergence very fast, in about a dozen of iterations. As EM-based approaches, they both guarantee the decrease of the loss function. Here, we should precise that slight oscillations might be observed for GraphEM as it solves a constrained minimization problem. The projection steps might break the monotonicity of the loss decrease, but this is not jeopardizing the convergence properties of the EM approach, and in practice the loss is rather stable. We notice the different behavior of the RMSE curves for both methods. The non regularized MLEM shows the typical noise amplification effect, decreasing first the RMSE and then increasing it. In contrast, the introduction of a suitable regularization strategy in GraphEM makes it avoid such undesirable phenomenon, and the RMSE evolution follows a stable decrease until converging to its final small value.

\begin{figure}
\begin{tabular}{cc}
\includegraphics[width = 0.2\textwidth]{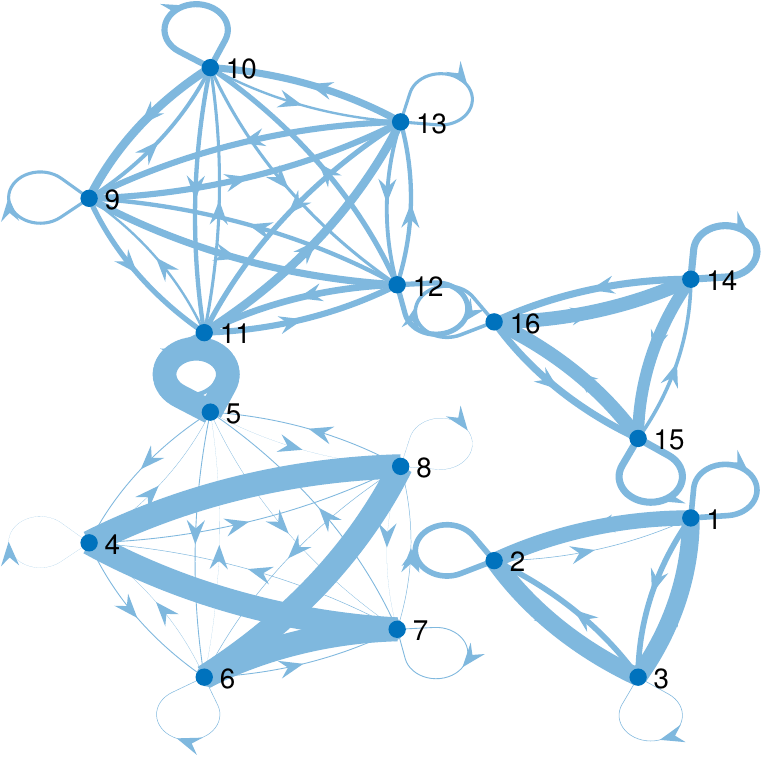} & \includegraphics[width = 0.2\textwidth]{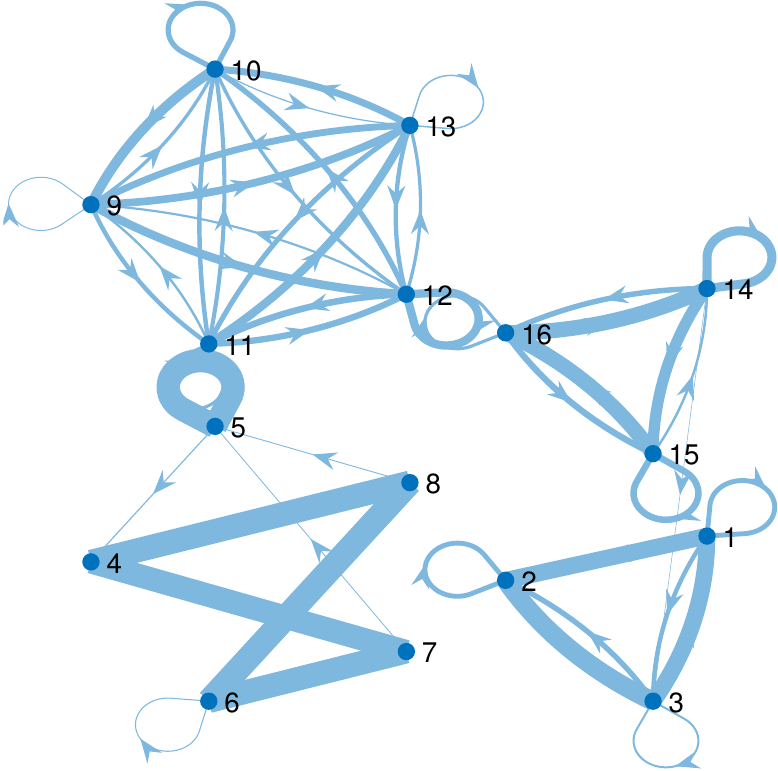} 
\end{tabular}
\vspace*{-0.2cm}
\caption{True graph (left) and GraphEM estimate (right) for dataset C.}
\label{fig:datasetC}
\end{figure}

\begin{figure}
\begin{tabular}{@{}c@{}c@{}}
\includegraphics[width = 0.24\textwidth]{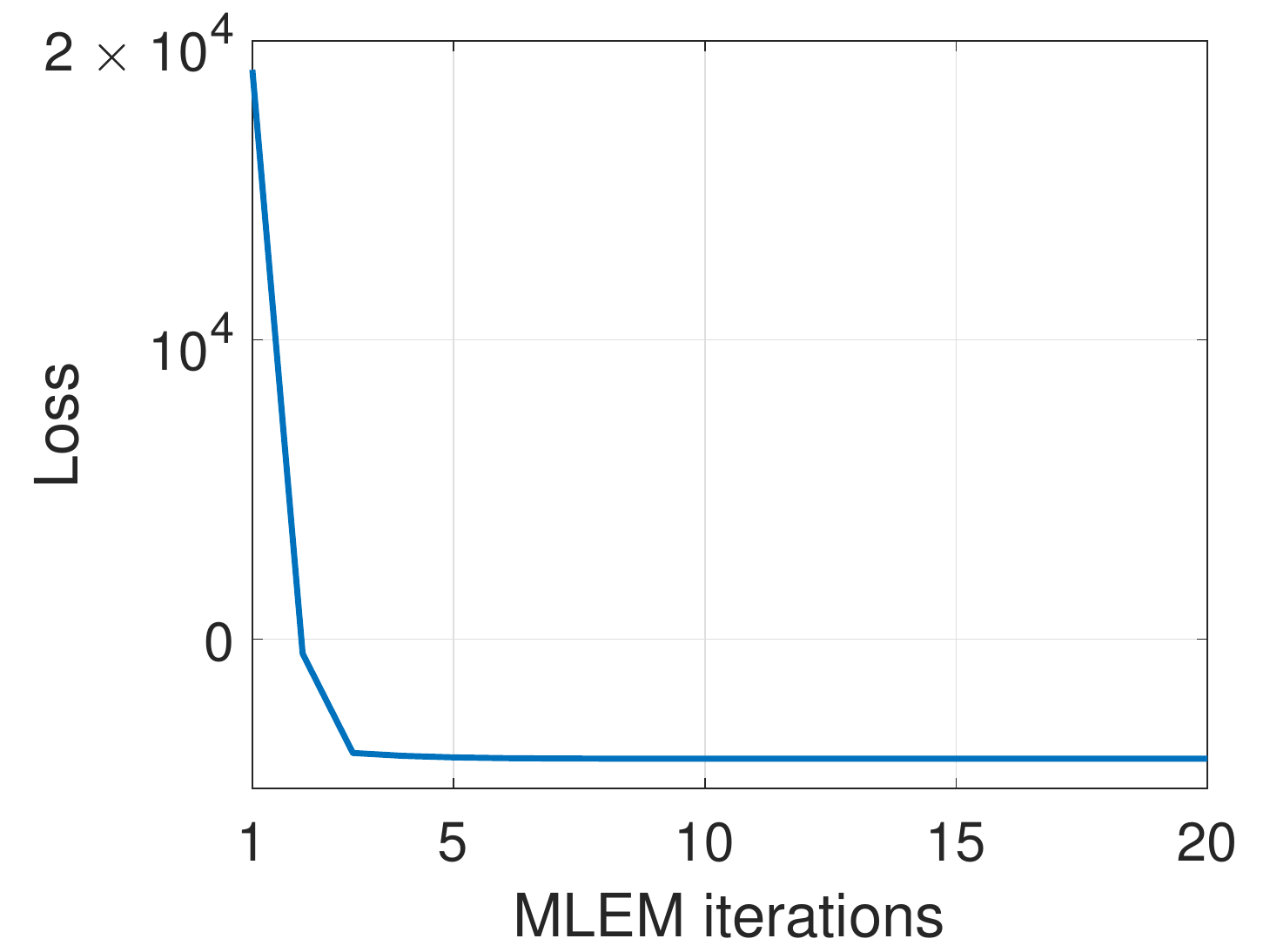} & \includegraphics[width = 0.24\textwidth]{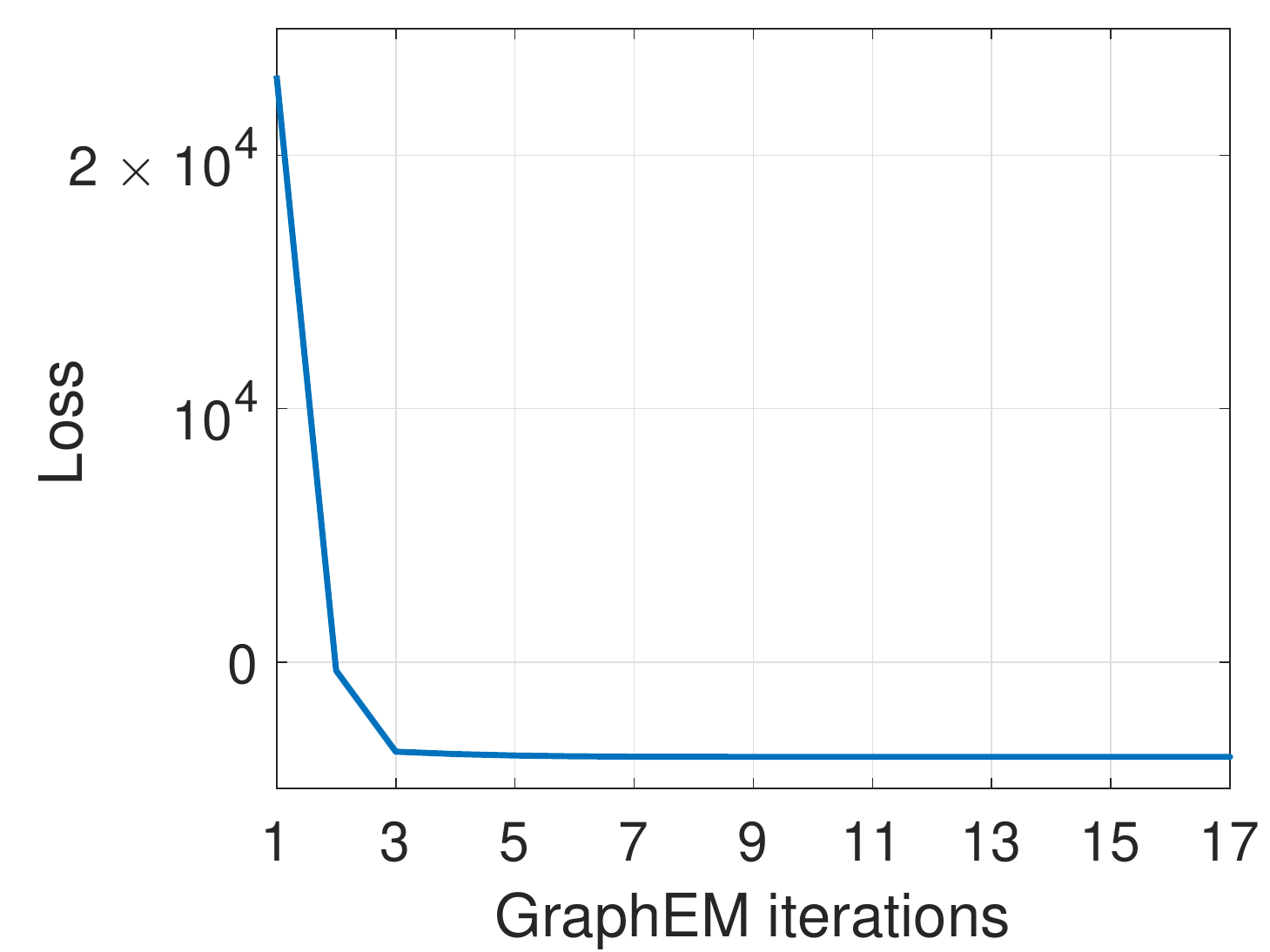}\\
 \includegraphics[width = 0.24\textwidth]{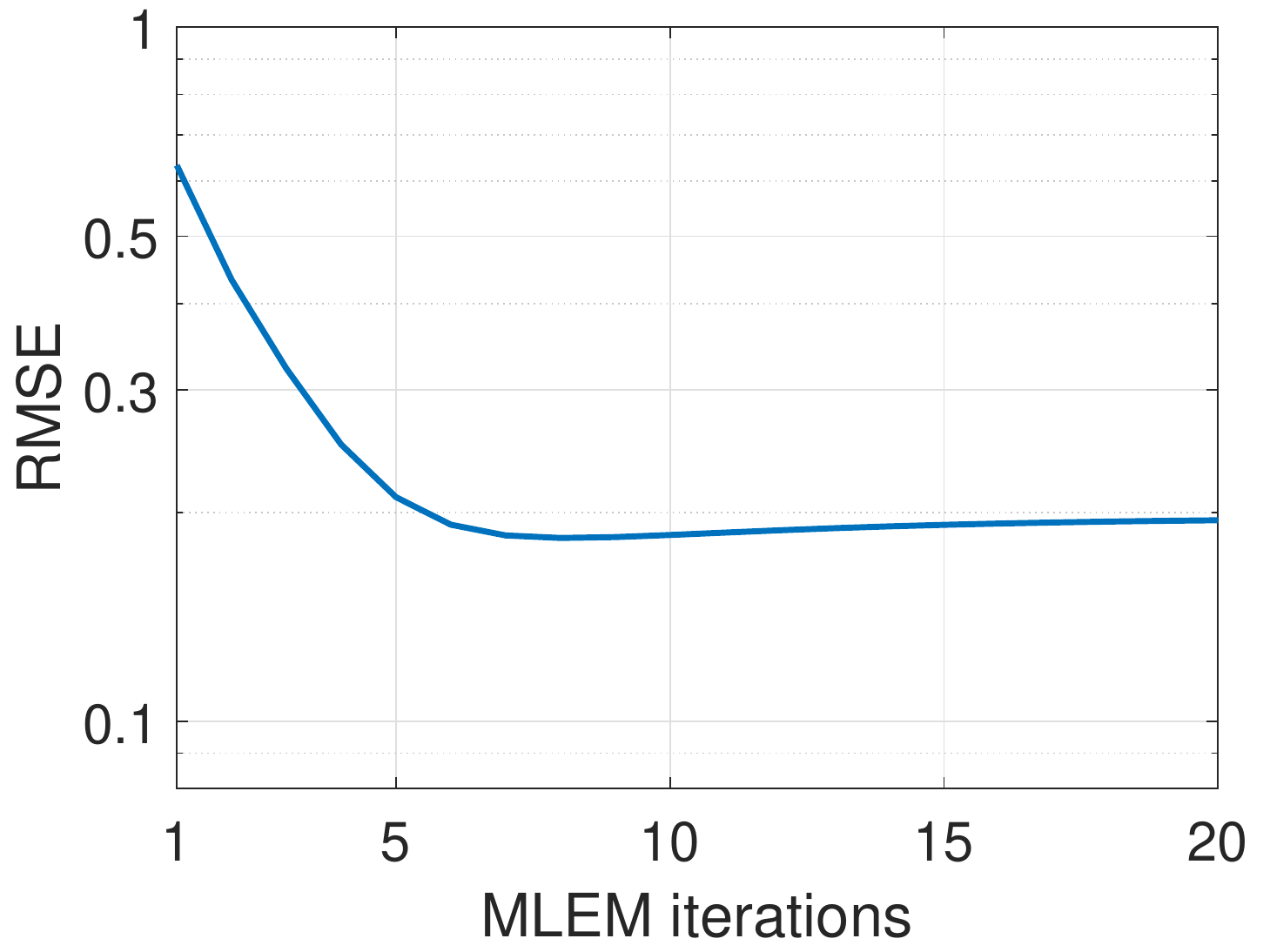} & \includegraphics[width = 0.24\textwidth]{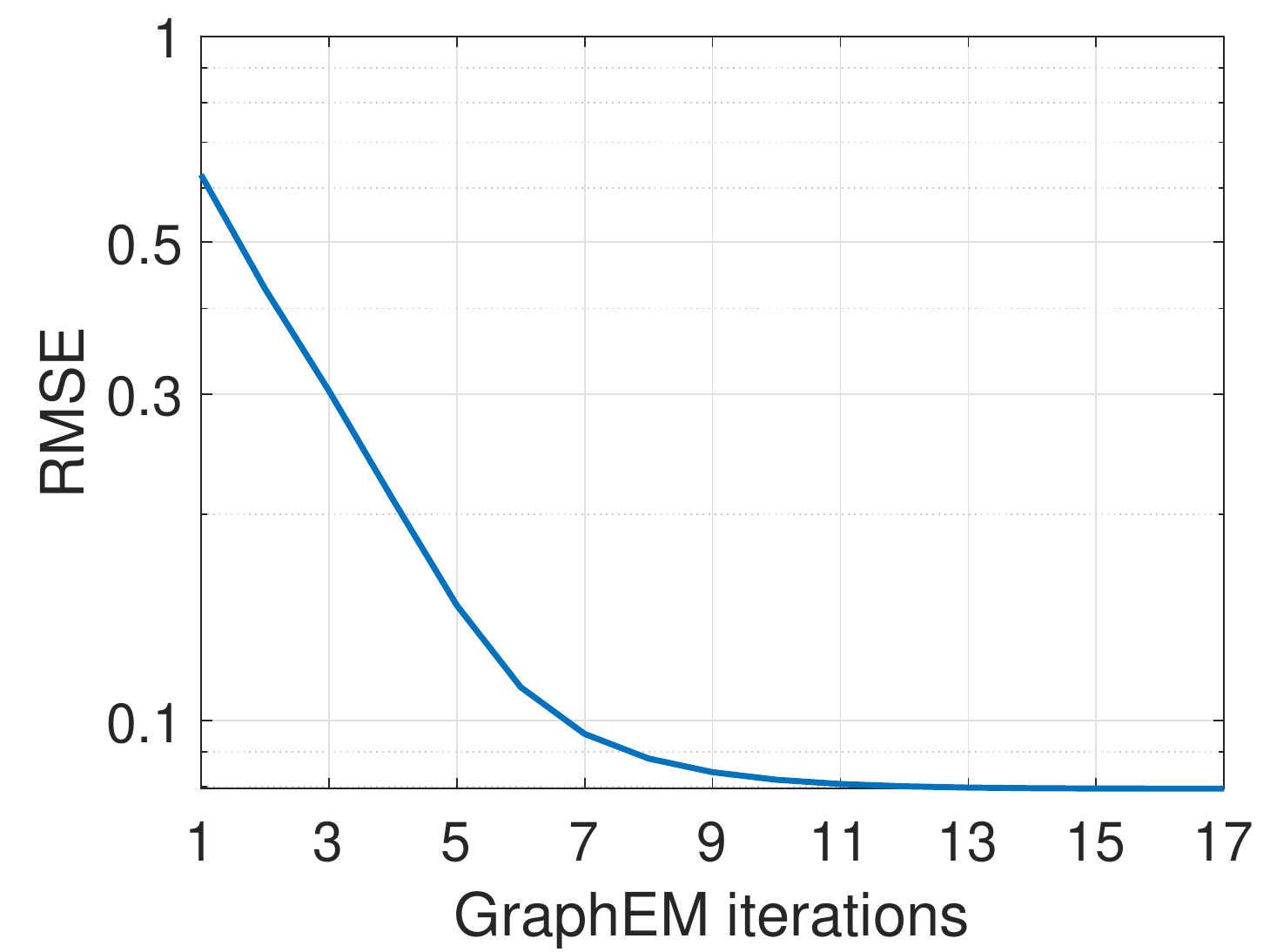}
\end{tabular}
\vspace*{-0.2cm}
\caption{Evolution of loss function (top) and RMSE score (bottom), for MLEM (left) and graphEM (right), when ran on a realization from dataset A.}
\label{fig:datasetAplots}
\end{figure}

\subsection{Wireless channel tracking}
\newcommand{\Nant}{L}
\newcommand{\Npil}{N_{\text{pil}}}
\newcommand{\channel}{\C}
\newcommand{\vect}{\text{vec}}
\newcommand{\pil}{\p}
\newcommand{\obsComp}{\z}
\newcommand{\noiseComplex}{\n}

{We consider a multi-input multi-output (MIMO) wireless communication system with fading \cite{elvira2021multiple}, where the (unknown) channel between the transmitter (TX) and the receiver (RX) must be tracked. The MIMO system is $\Nant \times \Nant$, although a different number of transmit and receive antennas is readily possible. 
At each time step, $\Npil$ $\Nant$-dimensional complex pilots, $\pil_k^{(i)}$, $i=1,\ldots,\Npil$, with $\Npil\geq 1$ and $64$-QAM symbols in each component, are transmitted between TX and RX through the complex channel $\channel_k \in \Complex^{\Nant \times \Nant}$. 
Therefore, the MIMO system with fading is modeled as
\begin{equation}   
 \obsComp_k^{(i)} = \channel_k \pil_k^{(i)} + \noiseComplex_k^{(i)},\label{eq_complex_mimo}
 \end{equation} 
with $k = 1,\ldots,K$ and $i = 1,\ldots,\Npil$, where $\noiseComplex_k^{(i)} \in \Complex^{\Nant}$, is distributed complex-normally with isotropic covariance that yields an $E_b/N_0=38$dBs. In order to express \eqref{eq_complex_mimo} as the (real-valued) observation model in \eqref{eq_model_obs}, we define $\x_k = [\text{Real}(\vect(\channel_k) ) ; \text{Imag}(\vect(\channel_k) )] \in \Real^{2L^2}$ (i.e., $N_x = 2L^2$) as the vectorized version of the complex channel. The real-valued observation vector corresponding to each pilot $i\in \{1,\ldots,\Npil\}$ is defined as $\y_k^{(i)} = [\text{Real}(\obsComp_k ) ; \text{Imag}(\obsComp_k )]$, so that $\y_k = [\y_k^{(1)},\ldots,\y_k^{(\Npil)}]^\top \in \Real^{N_y}$ with $N_y = 2L\Npil$. The observation matrix $\H_k\in \Real^{2L\Npil \times 2L^2}$ is a sparse matrix constructed from the real and imaginary components of all pilots, in such a way that the real-valued observation model in \eqref{eq_model_obs} is equivalent to \eqref{eq_complex_mimo}. The prior pdf of each entry of $\x_0$ is a standard normal distribution. We consider isotropic covariances \cred{$\Q = \sigma_\Q^2 {\Id}_{N_x}$} and $\R_k = \sigma_\R^2 {\Id}_{N_y}$, for every $k \in \{1,\ldots,K\}$, with $\sigma_\Q = 0.2$ and $\sigma_\R = 0.2$. \cblue{}We transmit $\Npil = 4$ at each time step, with $L=4$ transmit and receive antennas, hence $N_x = 32$ and $N_y = 32$.}

{The goal is estimating $\A \in \Real^{32 \times 32}$ 
by introducing sparse constraints motivated by the physical model, in such a way we can then do tracking of the channel with the estimated transition matrix.} We consider two datasets, obtained from two ground truth matrices $\A$. Dataset E relies on the tri-diagonal transition matrix:
\begin{multline}
(\forall (i,j) \in \{1,\ldots,32\}^2) \quad A(i,j) =\\
\begin{cases}
a & \quad \text{if} \quad i = j \; \text{or} \; i = j+16 \, \text{or} \; i + 16 = j,\\
0 &  \quad \text{otherwise},
\end{cases}
\end{multline}
with $a = 0.495$ set so that $\A$ belongs to the stability set \eqref{eq:setS} with $\delta = 0.99$. Dataset F uses the ground truth matrix $\A = \left[\begin{array}{cc} \mathbf{B} & \mathbf{B} \\ \mathbf{B} & \mathbf{B} \end{array}\right]$ with $\mathbf{B} \in \Real^{16 \times 16}$ a block diagonal matrix of 3 blocks with respective dimensions $(4,8,4)$. As in the example in subsection \ref{sec:synthetic}, randomly selected AR(1) matrices belonging to the stability set $\mathcal{S}$, are used to build the blocks of $\mathbf{B}$.

In both cases, observed data are simulated using \eqref{eq_model_obs} with $K = 200$ time steps. We compare the MLEM approach with the GraphEM algorithm, for the estimation of $\A$ from these data. In this example, we aim at exploring the robustness w.r.t. the regularization parameter in GraphEM when imposing both a stability and sparsity constraint, namely $\mathcal{L}_0 = f_2 + f_3$ with $f_2$ set as in our previous example, and $f_3 = \kappa \ell_{2,1}$ with weight parameter $\kappa>0$. The $\ell_{2,1}$ norm, as introduced in subsection \ref{sec_prior}, is a block-sparsity enhancing penalty. We preferred it to the $\ell_1$ norm in that example, as it better accounts for correlations between entries of $\A$ related to the same states in the complex domain. More precisely, following our notations from subsection \ref{sec_prior}, we set $B = L^4$ blocks, so that, for every $b \in \{1,\ldots,B\}$, and every $\Ab \in \mathbb{R}^{2 L^2 \times 2 L^2}$, we consider the $b$-th block of it as
{\small
\begin{multline*}
\mathbf{a}(b) = \\
\left[A({i,j}),A({i + L^2,j}),A({i,j+L^2}),A({i+L^2,j+L^2})\right]^\top \in \mathbb{R}^4,
\end{multline*}
}
with $(i,j) \in \{1,\ldots,L^2\}^2$ the index pair corresponding to the matrix position associated with the lexicographic index $b$. 
 
To that end, $\ell_{2,1}(\A)$ pairs the real and imaginary parts of the state at current and previous time state.  
Similar block-sparsity prior was used in \cite{Gueddari,Florescu} for processing complex-valued images. On both datasets, we run GraphEM algorithm with various weights $\kappa$ selected with the range $(0, 400]$, i.e., in a significantly wide range. We show two performance metrics to evidence the robustness and successful performance of GraphEM. First, in Fig. \ref{fig_wireless} (top), we show the relative mean square error (RMSE) in the estimation of the matrix $\A$ with respect to $\kappa$, either for dataset E (left) and dataset F (right). {We then design a more sophisticated BER analysis where we will track the channel and perform linear detection. Therefore, instead of plugging the true $\A$ to track the channel, we set the matrix estimates corresponding to MLEM or GraphEM, with different values of $\kappa$. More precisely, at each time step, we track the channel $\channel_k$ in the same model described above. The difference is that now we run the Kalman filter setting the estimated $\A$ from each corresponding algorithm. For each  channel use, we transmit $\Npil = 4$ pilots (for tracking purposes) and $500$ (unknown) symbols for evaluating the BER performance under the estimated $\A$ matrix of each algorithm. We decode the transmitted symbols by using the MMSE detector  \cite{jiang2011performance}. We run this testing phase for $10^4$ time steps, for ensuring a sufficiently averaged BER metric \cite{elvira2019multiple}. The BER, as a function of the parameter $\kappa$, is shown in Fig. \ref{fig_wireless} (bottom) for both MLEM and GraphEM approaches running on both dataset E (left) and dataset F (right). In all the four plots, we can see that GraphEM outperforms the MLEM approach, obtaining the best performance for a value around $\kappa =100$. We can also see, that the performance is good for a wide range of $\kappa$ values with an asymmetric behavior: larger values of $\kappa$ still retain the advantage of using GraphEM in this example. This shows the stability of GraphEM model to the setting of~$\kappa$.}

 \begin{figure}
\centering
\begin{tabular}{@{}c@{}c@{}}
\includegraphics[width = 0.24\textwidth]{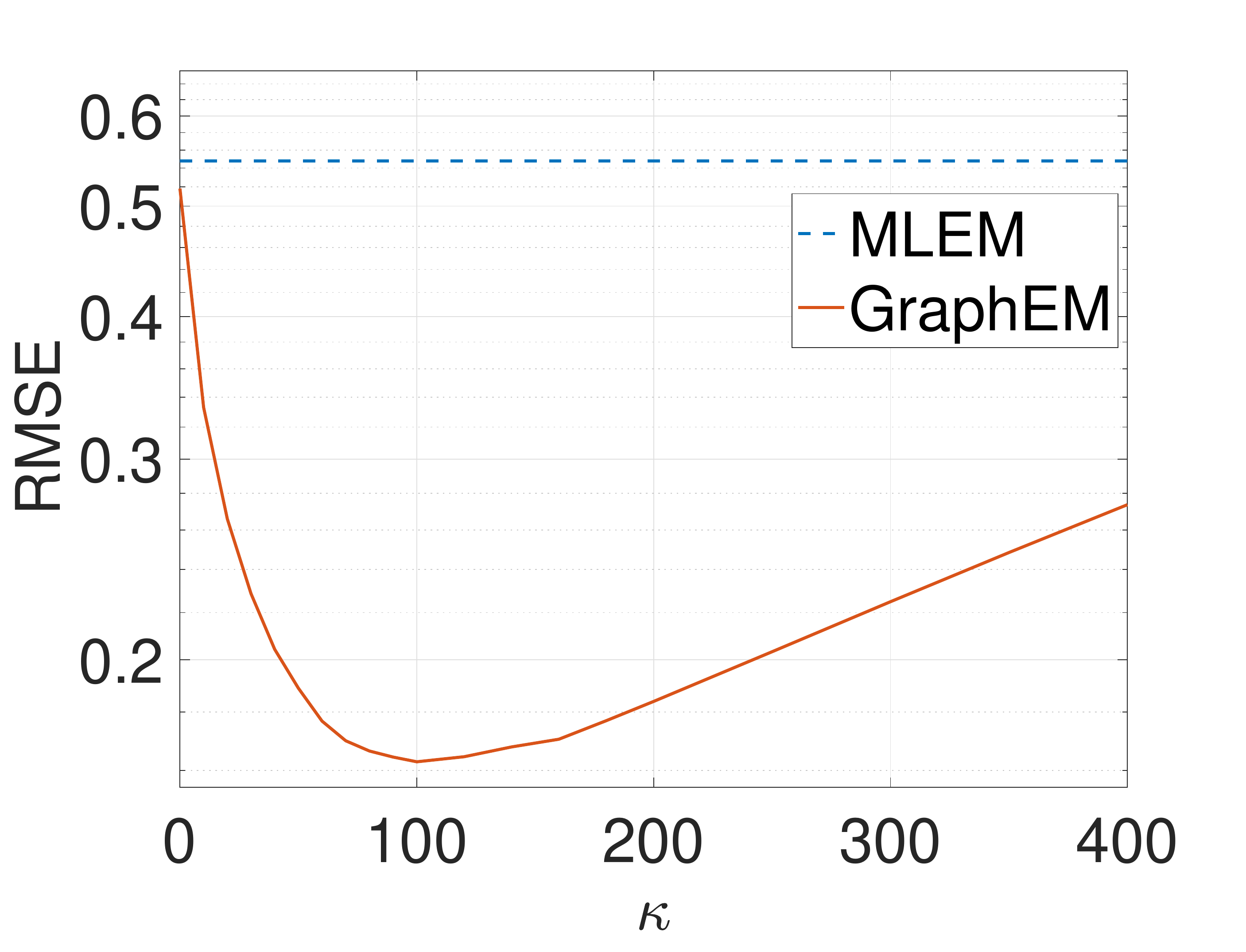}  & \includegraphics[width = 0.23\textwidth]{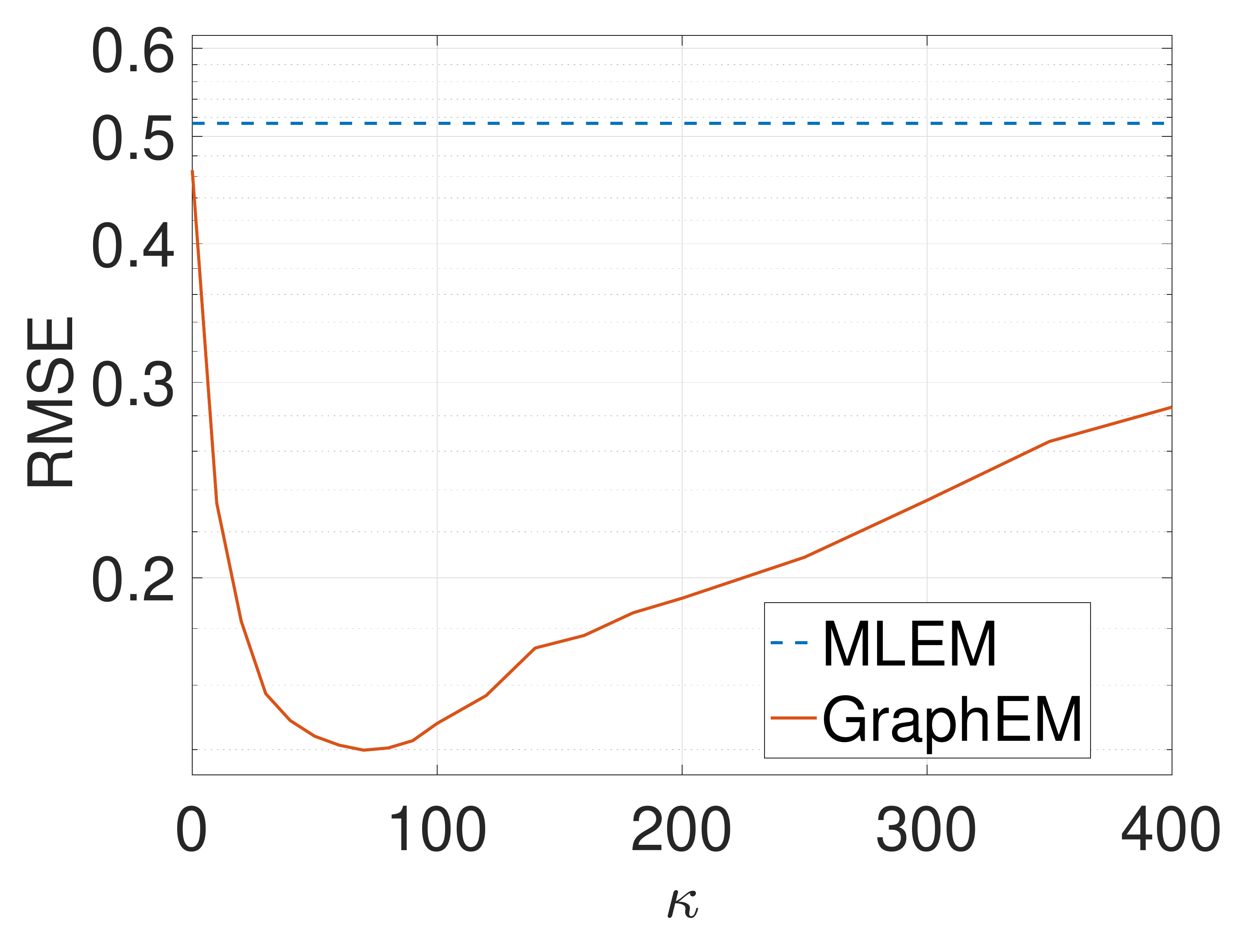} \\
\includegraphics[width = 0.24\textwidth]{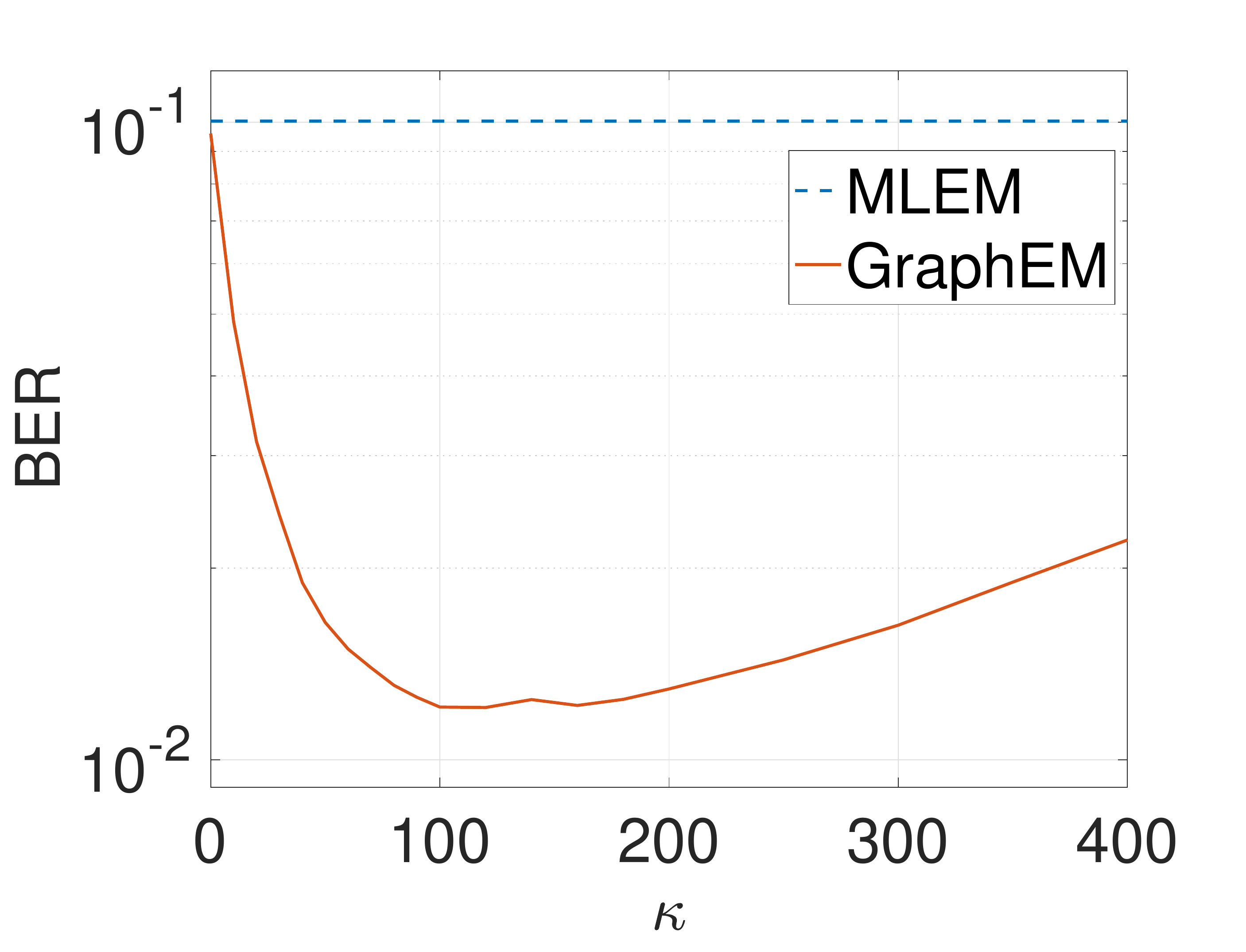}  & \includegraphics[width = 0.23\textwidth]{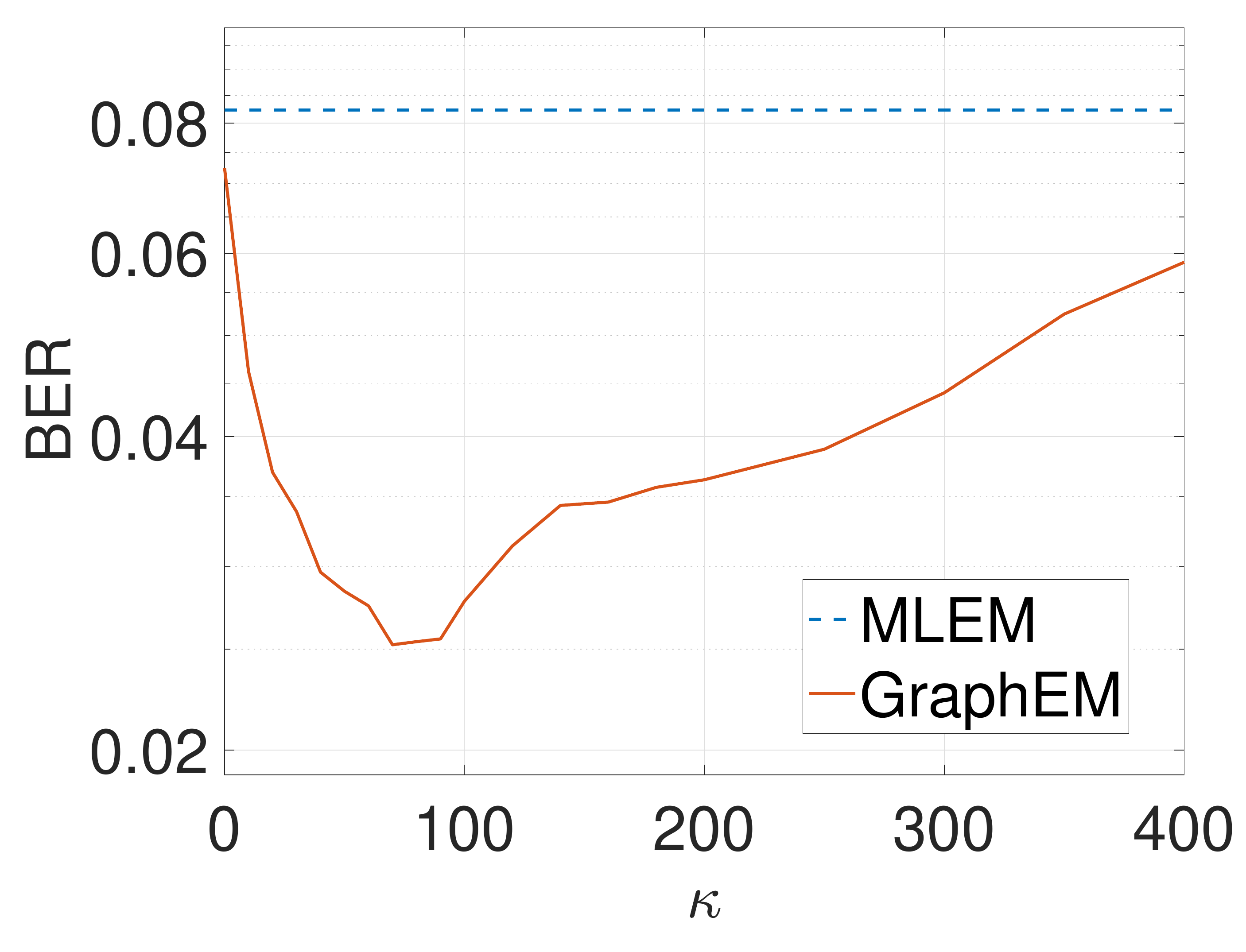} 
\end{tabular}
\caption{\textbf{Wireless channel tracking.} {Performance of GraphEM algorithm in terms of (top) the relative mean square error (RMSE) in the estimation of $\A$ and (bottom) bit error rate (BER) as function of $\kappa$, the weight associated to the $\ell_{2,1}$ norm, on dataset E (left) and F (right).}  
}
\label{fig_wireless}
\end{figure}

\section{Conclusion}
\label{sec:conclusion}
 
%
In this paper, we have proposed a novel methodological framework, called GraphEM, to estimate parameters in the linear-Gaussian state-space model (LG-SSM) by introducing available application-dependent prior knowledge. 
While the methodology is generic to allow for the MAP estimate of all LG-SSM model parameters, we develop further our method for the estimation of the transition matrix. 
Our novel approach interprets this matrix as the adjacency matrix of a directed graph, encoding the Markovian dependencies in the evolution of the multi-variate state. This interpretation has some ties with Granger causality. 
{We then propose GraphEM for estimating this matrix jointly with the inference of the sequence of hidden states.} GraphEM is \cblue{a convergent} expectation-maximization (EM) methodology which incorporates a novel consensus-based implementation of a primal-dual proximal convex optimization solver for the M-step, enabling an efficient incorporation of sophisticated priors on the graph. 
Numerical results illustrate the great performance of the method.
The novel interpretation, the solid theoretical guarantees, and the good performance of GraphEM pave the way for novel advances. For example, we have considered in our numerical examples several penalties on the graph structure, such as stability of the hidden process and block-sparsity enhancing priors that allow for simple and interpretable graphs. The versatility of our method allows to introduce other priors to target specific applications.

{\color{black}

\appendices

\section{Useful expectations involving Normal pdfs}
\label{appendix_i1}
Let us consider a random vector $\mathbf{X}$ in dimension $n \geq 1$, following a multivariate Gaussian distribution with mean $\widetilde{\x}$ and covariance matrix $\widetilde{\P}$:
\begin{equation}
\X \sim \mathcal{N}(\widetilde{\x},\widetilde{\P}),
\end{equation}
with $\widetilde{\x} \in \mathbb{R}^n$, $\widetilde{\P} \in \mathbb{R}^{n \times n}$ symmetric positive definite. We are interested in computing
\begin{equation}
\mathbb{E}\{(\X-\mub)^\top \Sigmab^{-1} (\X-\mub)\},
\end{equation}
for some $\mub \in \mathbb{R}^n$ and $\Sigmab \in \mathbb{R}^{n \times n}$ symmetric positive definite. We have
\footnotesize{
\begin{align}
\mathbb{E}& \{(\X-\mub)^\top \Sigmab^{-1} (\X-\mub)\} \nonumber \\
& = \mathbb{E}\{(\X-\widetilde{\x})^\top \Sigmab^{-1} (\X-\widetilde{\x})\}
+ (\mub-\widetilde{\x})^\top \Sigmab^{-1} (\mub-\widetilde{\x}) \label{eq:trace1}\\
& = \text{tr} \left( \mathbb{E}\{(\X-\widetilde{\x}) (\X-\widetilde{\x})^\top \Sigmab^{-1}\}\right)
+ (\mub-\widetilde{\x})^\top \Sigmab^{-1} (\mub-\widetilde{\x})  \label{eq:trace2}\\
& = \text{tr} \left( \mathbb{E}\{(\X-\widetilde{\x}) (\X-\widetilde{\x})^\top\} \Sigmab^{-1}\right)
+ (\mub-\widetilde{\x})^\top \Sigmab^{-1} (\mub-\widetilde{\x})\\
& =  \text{tr} (\widetilde{\P} \Sigmab^{-1}) + \text{tr}( 
(\mub-\widetilde{\x})^\top \Sigmab^{-1} (\mub-\widetilde{\x})) \nonumber. 
\end{align}}}
\cblue{Finally,}
\footnotesize{
\begin{align}
\mathbb{E}& \{(\X-\mub)^\top \Sigmab^{-1} (\X-\mub)\}  
= \text{tr} (\Sigmab^{-1} \widetilde{\P} ) + \text{tr}(\Sigmab^{-1} (\mub-\widetilde{\x})  (\mub-\widetilde{\x})^\top) \nonumber \\
& = \text{tr} (\Sigmab^{-1} (\widetilde{\P} + (\mub-\widetilde{\x})  (\mub-\widetilde{\x})^\top)) \label{eqappA}.
\end{align}
}
\normalsize

\section{E-step calculations}
\label{appendix_i2}
Here, we explicit the end of the calculations needed for $q(\Ab;\Ab^{(i)})$. We recall that  $\widetilde{\Ab} = [\Id_{N_x} , - \Ab]$. Let $k \in \{1,\ldots,K\}$. Then,
\begin{equation}
\widetilde{\Ab}^\top \cred{\Qb}^{-1} \widetilde{\Ab} = \left[\begin{array}{cc} \cred{\Qb}^{-1} & - \cred{\Qb}^{-1} \Ab\\  \Ab^\top \cred{\Qb}^{-1} & \Ab^\top \cred{\Qb}^{-1} \Ab \end{array}\right].
\end{equation}
Morever,
\begin{multline}
\meanSmooth_{k:k-1} (\meanSmooth_{k:k-1})^\top = \left[\begin{array}{cc} \meanSmooth_{k} (\meanSmooth_{k})^\top  & \meanSmooth_{k-1} (\meanSmooth_{k})^\top \\  \meanSmooth_{k} (\meanSmooth_{k-1})^\top & \meanSmooth_{k-1} (\meanSmooth_{k-1})^\top \end{array}\right].
\end{multline}
Thus,
\begin{align}
& \text{tr}\left(\widetilde{\Ab}^\top \cred{\Qb}^{-1} \widetilde{\Ab} (\covSmooth_{k:k-1} + \meanSmooth_{k:k-1} (\meanSmooth_{k:k-1})^\top)\right)\\
& = \text{tr}\left(
\left[\begin{array}{cc} \cred{\Qb}^{-1} & - \cred{\Qb}^{-1} \Ab\\  \Ab^\top \cred{\Qb}^{-1} & \Ab^\top \cred{\Qb}^{-1} \Ab \end{array}\right]\right.\nonumber\\
& \left.
\quad \times
\left(\left[\begin{array}{cc} \covSmooth_{k} & \covSmooth_{k} \Gb_{k-1}^\top \\  \Gb_{k-1} \covSmooth_{k} & \covSmooth_{k-1} \end{array}\right] \right.\right. \nonumber\\
& \qquad \left.\left.
+ \left[\begin{array}{cc} \meanSmooth_{k} (\meanSmooth_{k})^\top  & \meanSmooth_{k-1} (\meanSmooth_{k})^\top \\  \meanSmooth_{k} (\meanSmooth_{k-1})^\top & \meanSmooth_{k-1} (\meanSmooth_{k-1})^\top \end{array}\right]\right)
\right). \label{eq:traceI2}
\end{align}
In order to limit computations, we can make use of the fact that
\begin{equation}
\text{tr}\left( \left[\begin{array}{cc} A  & B \\  C & D \end{array}\right]\right) = \text{tr}(A) + \text{tr}(D). \label{eq:traceprop}
\end{equation}
{
Using  \eqref{eq:traceprop} and the additivity of the trace in \eqref{eq:traceI2} leads to equality \eqref{eq:appB}(a). Then, using the permutation property of the trace yields the equality \eqref{eq:appB}(b):
\normalsize{
\begin{align}
& \text{tr}\left(\widetilde{\Ab}^\top {\Qb}^{-1} \widetilde{\Ab} (\covSmooth_{k:k-1} + \meanSmooth_{k:k-1} (\meanSmooth_{k:k-1})^\top)\right) \nonumber\\
& \overset{(a)}{=} \text{tr}\left( 
{\Qb}^{-1} (\covSmooth_k + \meanSmooth_{k} (\meanSmooth_{k})^\top)\right) \nonumber\\
& \quad +\text{tr}\left(- {\Qb}^{-1} \Ab   (\Gb_{k-1} {\covSmooth_k} + \meanSmooth_{k-1} (\meanSmooth_{k})^\top )  \right) 
\nonumber \\
& \quad
+ \text{tr}\left(- \Ab^\top  {\Qb}^{-1} ( \covSmooth_k \Gb_{k-1}^\top  + \meanSmooth_{k} (\meanSmooth_{k-1})^\top ) \right) \nonumber \\
& \quad +  \text{tr}\left(\Ab^\top {\Qb}^{-1} \Ab ( \covSmooth_{k-1} + \meanSmooth_{k-1} (\meanSmooth_{k-1})^\top)
\right) \notag \\
& \overset{(b)}{=} \text{tr}\left(  
 {\Qb}^{-1} (\covSmooth_k + \meanSmooth_{k} (\meanSmooth_{k})^\top)\right) \nonumber \\
& \quad +  \text{tr}\left(- {\Qb}^{-1} (\covSmooth_k \Gb_{k-1}^\top  + \meanSmooth_{k} (\meanSmooth_{k-1})^\top ) \Ab^\top  \right) \nonumber\\
& \quad + \text{tr}\left(- {\Qb}^{-1} \Ab   (\Gb_{k-1}  {\covSmooth_k} + \meanSmooth_{k-1} (\meanSmooth_{k})^\top )  \right)  
 \nonumber\\
& \quad + \text{tr}\left({\Qb}^{-1} \Ab ( \covSmooth_{k-1} + \meanSmooth_{k-1} (\meanSmooth_{k-1})^\top) \Ab^\top \right). \label{eq:appB}
\end{align}
}
}
{\normalsize Finally,
 \footnotesize{
\begin{align}
& \text{tr}\left(\widetilde{\Ab}^\top {\Qb}^{-1} \widetilde{\Ab} (\covSmooth_{k:k-1} + \meanSmooth_{k:k-1} (\meanSmooth_{k:k-1})^\top)\right) \nonumber\\
& = \text{tr}\left( 
 {\Qb}^{-1} (\covSmooth_k + \meanSmooth_{k} (\meanSmooth_{k})^\top  -   (\covSmooth_k \Gb_{k-1}^\top  + \meanSmooth_{k} (\meanSmooth_{k-1})^\top ) \Ab^\top \right. \nonumber \\
& \quad
 \left.- \Ab   (\Gb_{k-1} \covSmooth_k  {+}  \meanSmooth_{k-1} (\meanSmooth_{k})^\top ) + \Ab ( \covSmooth_{k-1} + \meanSmooth_{k-1} (\meanSmooth_{k-1})^\top) \Ab^\top ) \right). \label{eqappb}
\end{align}
}
}
\normalsize

\bibliographystyle{IEEEtr}

\end{document}